\documentclass[11pt,english]{article}
\usepackage[T1]{fontenc}
\usepackage[latin9]{inputenc}
\usepackage[a4paper]{geometry}
\usepackage{array}
\usepackage{url}
\usepackage{amsmath}
\usepackage{graphicx}
\usepackage{setspace}
\usepackage{amssymb}
\makeatletter

\newcommand{\noun}[1]{\textsc{#1}}
\providecommand{\tabularnewline}{\\}

\newcommand{\lyxaddress}[1]{
\par {\raggedright #1
\vspace{1.4em}
\noindent\par}
}

\usepackage{float}
\usepackage{amsfonts}
\usepackage{amsthm}
\usepackage{psfrag}
\usepackage{epsfig}
\theoremstyle{plain}\newtheorem{prop}{Proposition}
\numberwithin{equation}{section}

\date{}
\renewcommand{\[}{\begin{equation}}
\renewcommand{\]}{\end{equation}}
\makeatother

\usepackage{babel}

\begin{document}

\title{Simultaneous model-based clustering and\\
visualization in the Fisher discriminative subspace}

\author{Charles \noun{Bouveyron$^{1}$ \&} Camille \noun{Brunet$^{2}$}}

\maketitle

\lyxaddress{\begin{center}
$^{1}$ Laboratoire SAMM, EA 4543, Université Paris 1 Panthéon-Sorbonne\\
90 rue de Tolbiac, 75013 Paris, France\\
Email: \url{charles.bouveyron@univ-paris1.fr}\medskip\\
$^{2}$ IBISC, TADIB, FRE CNRS 3190, Université d'Evry Val d'Essonne\\
40 rue de Pelvoux, CE 1455, 91020 Evry Courcouronnes, France\\
Email: \url{camille.brunet@ibisc.univ-evry.fr}
\par\end{center}}
\begin{abstract}
Clustering in high-dimensional spaces is nowadays a recurrent problem
in many scientific domains but remains a difficult task from both
the clustering accuracy and the result understanding points of view.
This paper presents a discriminative latent mixture (DLM) model which
fits the data in a latent orthonormal discriminative subspace with
an intrinsic dimension lower than the dimension of the original space.
By constraining model parameters within and between groups, a family
of 12 parsimonious DLM models is exhibited which allows to fit onto
various situations. An estimation algorithm, called the Fisher-EM
algorithm, is also proposed for estimating both the mixture parameters
and the discriminative subspace. Experiments on simulated and real
datasets show that the proposed approach performs better than existing
clustering methods while providing a useful representation of the
clustered data. The method is as well applied to the clustering of
mass spectrometry data.\medskip
\end{abstract}

\section{Introduction}

In many scientific domains, the measured observations are nowadays
high-dimensional and clustering such data remains a challenging problem.
Indeed, the most popular clustering methods, which are based on the
mixture model, show a disappointing behavior in high-dimensional spaces.
They suffer from the well-known \textit{curse of dimensionality}~\cite{Bellman57}
which is mainly due to the fact that model-based clustering methods
are over-parametrized in high-dimensional spaces. Furthermore, in
several applications such as mass spectrometry or genomics, the number
of available observations is small compared to the number of variables
and such a situation increases the difficulty of the problem.

Hopefully, since the dimension of observed data is usually higher
than their intrinsic dimension, it is theoretically possible to reduce
the dimension of the original space without loosing any information.
Therefore, dimension reduction methods are traditionally used before
the clustering step. Feature extraction methods such as principal
component analysis (PCA) or feature selection methods are very popular.
However, these approaches of dimension reduction do not consider the
classification task and provide a sub-optimal data representation
for the clustering step. Indeed, dimension reduction methods imply
an information loss which could be discriminative. Only few approaches
combine dimension reduction with the classification aim but, unfortunately,
those approaches are all supervised methods. Fisher discriminant analysis
(FDA) (see Chap.~4 of~\cite{Hastie01}) is one of them in the supervised
classification framework. FDA is a powerful tool for finding the subspace
which best discriminates the classes and reveals the structure of
the data. This subspace is spanned by the discriminative axes which
maximize the ratio of the between class and the within class variances.

To avoid dimension reduction, several subspace clustering methods
have been proposed in the past few years to model the data of each
group in low-dimensional subspaces. These methods turned out to be
very efficient in practice. However, since these methods model each
group in a specific subspace, they are not able to provide a global
visualization of the clustered data which could be helpful for the
practitioner. Indeed, the clustering results of high-dimensional data
are difficult to understand without a visualization of the clustered
data. In addition, in scientific fields such as genomics or economics,
original variables have an actual meaning and the practitioner could
be interested in interpreting the clustering results according to
the variable meaning.

In order to both overcome the curse of dimensionality and improve
the understanding of the clustering results, this work proposes a
method which adapts the traditional mixture model for modeling and
classifying data in a latent discriminative subspace. For this, the
proposed discriminative latent mixture (DLM) model combines the model-based
clustering goals with the discriminative criterion introduced by Fisher.
The estimation procedure proposed in this paper and named Fisher-EM
has three main objectives: firstly, it aims to improve clustering
performances with the use of a discriminative subspace, secondly,
it avoids estimation problems linked to high dimensions through model
parsimony and, finally, it provides a low-dimensional discriminative
representation of the clustered data.

The reminder of this manuscript has the following organization. Section~\ref{sec:Related-works}
reviews the problem of high-dimensional data clustering and existing
solutions. Section~\ref{sec:Model-DLM} introduces the discriminative
latent mixture model and its submodels. The link with existing approaches
is also discussed in Section~\ref{sec:Model-DLM}. Section~\ref{sec:Parameter-Estimation}
presents an EM-based procedure, called Fisher-EM, for estimating the
parameters of the DLM model. Initialization, model selection and convergence
issues are also considered in Section~\ref{sec:Parameter-Estimation}.
In particular, the convergence of the Fisher-EM algorithm has been
proved in this work only for one of the DLM models and the convergence
for the other models should be investigated. In Section~\ref{sec:Experimental-results},
the Fisher-EM algorithm is compared to existing clustering methods
on simulated and real datasets. Section~\ref{sec:Application-to-mass-spectrometry}
presents the application of the Fisher-EM algorithm to a real-world
clustering problem in mass-spectrometry imaging. Some concluding remarks
and ideas for further works are finally given in Section~\ref{sec:Conclusion}.

\section{Related works\label{sec:Related-works}}

Clustering is a traditional statistical problem which aims to divide
a set of observations $\{y_{1},\dots,y_{n}\}$ described by $p$ variables
into $K$ homogeneous groups. The problem of clustering has been widely
studied for years and the reader could refer to~\cite{Raftery02,Jain99}
for reviews on the clustering problem. However, the interest in clustering
is still increasing since more and more scientific fields require
to cluster high-dimensional data. Moreover, such a task remains very
difficult since clustering methods suffer from the well-known \textit{curse
of dimensionality}~\cite{Bellman57}. Conversely, the \textit{empty
space phenomenon}~\cite{Scott83}, which refers to the fact that
high-dimensional data do not fit the whole observation space but live
in low-dimensional subspaces, gives hope to efficiently classify high-dimensional
data. This section firstly reviews the framework of model-based clustering
before exposing the existing approaches to deal with the problem of
high dimension in clustering.

\subsection{Model-based clustering and high-dimensional data}

\label{s:modclus}

Model-based clustering, which has been widely studied by~\cite{Raftery02,McLachlan00}
in particular, aims to partition observed data into several groups
which are modeled separately. The overall population is considered
as a mixture of these groups and most of time they are modeled by
a Gaussian structure. By considering a dataset of $n$ observations
$\{y_{1},\dots,y_{n}\}$ which is divided into $K$ homogeneous groups
and by assuming that the observations $\{y_{1},...,y_{n}\}$ are independent
realizations of a random vector $Y\in\mathbb{R}^{p}$ , the mixture
model density is then:

\[
f(y)=\sum_{k=1}^{K}\pi_{k}f(y;\theta_{k}),\]
where $f(.;\theta_{k})$ is often the multivariate Gaussian density
$\phi(.;\mu_{k},\Sigma_{k})$ parametrized by a mean vector $\mu_{k}$
and a covariance matrix $\Sigma_{k}$ for the $k$th component. Unfortunately,
model-based clustering methods show a disappointing behavior when
the number of observations is small compared to the number of parameters
to estimate. Indeed, in the case of the full Gaussian mixture model,
the number of parameters to estimate is a function of the square of
the dimension $p$ and the estimation of this potentially large number
of parameters is consequently difficult with a small dataset. In particular,
when the number of observations $n$ is of the same order than the
number of dimensions $p$, most of the model-based clustering methods
have to face numerical problems due to the ill-conditioning of the
covariance matrices. Furthermore, it is not possible to use the full
Gaussian mixture model without restrictive assumptions for clustering
a dataset for which $n$ is smaller than $p$. Indeed, for clustering
such data, it would be necessary to invert $K$ covariance matrices
which would be singular. To overcome these problems, several strategies
have been proposed in the literature among which dimension reduction
and subspace clustering.

\subsection{Dimension reduction and clustering}

Earliest approaches proposed to overcome the problem of high dimension
in clustering by reducing the dimension before using a traditional
clustering method. Among the unsupervised tools of dimension reduction,
PCA~\cite{Jolliffe86} is the traditional and certainly the most
used technique for dimension reduction. It aims to project the data
on a lower dimensional subspace in which axes are built by maximizing
the variance of the projected data. Non-linear projection methods
can also be used. We refer to~\cite{Verleysen05} for a review on
these alternative dimension reduction techniques. In a similar spirit,
the generative topographic mapping (GTM)~\cite{Bishop98} finds a
non linear transformation of the data to map them on low-dimensional
grid. An other way to reduce the dimension is to select relevant variables
among the original variables. This problem has been recently considered
in the clustering context by~\cite{Boutemedjet09} and~\cite{Law04}.
In~\cite{Raftery06} and~\cite{Maugis09}, the problem of feature
selection for model-based clustering is recasted as a model selection
problem. However, such approaches remove variables and consequently
information which could have been discriminative for the clustering
task.

\subsection{Subspace clustering}

In the past few years, new approaches focused on the modeling of each
group in specific subspaces of low dimensionality. Subspace clustering
methods can be split into two categories: heuristic and probabilistic
methods. Heuristic methods use algorithms to search for subspaces
of high density within the original space. On the one hand, bottom-up
algorithms use histograms for selecting the variables which best discriminate
the groups. The Clique algorithm~\cite{Agrawal98} was one of the
first bottom-up algorithms and remains a reference in this family
of methods. On the other hand, top-down algorithms use iterative techniques
which start with all original variables and remove at each iteration
the dimensions without groups. A review on heuristic methods is available
in~\cite{Parsons04}. Conversely, probabilistic methods assume that
the data of each group live in a low-dimensional latent space and
usually model the data with a generative model. Earlier strategies~\cite{Rubin82}
are based on the factor analysis model which assumes that the latent
space is related with the observation space through a linear relationship.
This model was recently extended in~\cite{Baek2009,McLachlan2003}
and yields in particular the well known mixture of probabilistic principal
component analyzers~\cite{Tipping99}. Recent works~\cite{Bouveyron07,McNicholas2008}
proposed two families of parsimonious and regularized Gaussian models
which partially encompass previous approaches. All these techniques
turned out to be very efficient in practice to cluster high-dimensional
data. However, despite their qualities, these probabilistic methods
mainly consider the clustering aim and do not take enough into account
the visualization and understanding aspects.

\subsection{From Fisher's theory to discriminative clustering}

In the case of supervised classification, Fisher poses, in his precursor
work~\cite{Fisher36}, the problem of the discrimination of three
species of iris described by four measurements. The main goal of Fisher
was to find a linear subspace that separates the classes according
to a criterion (see~\cite{Duda} for more details). For this, Fisher
assumes that the dimensionality $p$ of the original space is greater
than the number $K$ of classes. Fisher discriminant analysis looks
for a linear transformation $U$ which projects the observations in
a discriminative and low dimensional subspace of dimension $d$ such
that the linear transformation $U$ of dimension ${p\times d}$ aims
to maximize a criterion which is large when the between-class covariance
matrix~($S_{B}$) is large and when the within-covariance matrix~($S_{W}$)
is small. Since the rank of $S_{B}$ is at most equal to $K-1$, the
dimension $d$ of the discriminative subspace is therefore at most
equal to $K-1$ as well. Four different criteria can be found in the
literature which satisfy such a constraint (see~\cite{Fukunaga90}
for a review). The criterion which is traditionally used is: \begin{equation}
J(U)=\mathrm{trace}((U^{t}S_{W}U)^{-1}U^{t}S_{B}U),\label{eq:criter_Fisher}\end{equation}
 where $S_{W}=\frac{1}{n}\sum_{k=1}^{K}\sum_{i\in C_{k}}(y_{i}-m_{k})(y_{i}-m_{k})^{t}$
and $S_{B}=\frac{1}{n}\sum_{k=1}^{K}n_{k}(m_{k}-\bar{y})(m_{k}-\bar{y})^{t}$
are respectively the within and the between covariance matrices, $m_{k}=\frac{1}{n_{k}}\sum_{i\in C_{k}}^{K}y_{i}$
is the empirical mean of the observed column vector $y_{i}$ in the
class $k$ and $\bar{y}=\frac{1}{n}\sum_{k=1}^{K}n_{k}m_{k}$ is the
mean column vector of the observations. The maximization of criterion~\eqref{eq:criter_Fisher}
is equivalent to the generalized eigenvalue problem~\cite{Krzanowski03}
$\left(S_{W}^{-1}S_{B}-\lambda I_{p}\right)U=0$ and the classical
solution of this problem is the eigenvectors associated to the $d$
largest eigenvalues of the matrix~$S_{W}^{-1}S_{B}$. From a practical
point of view, this optimization problem can also be solved using
generalized eigenvalue solvers~\cite{Golub91} in order to avoid
numerical problems when $S_{W}$ is ill-conditioned. Once the discriminative
axes determined, linear discriminant analysis (LDA) is usually applied
to classify the data into this subspace. The optimization of the Fisher
criterion supposes the non-singularity of the matrix $S_{W}$ but
it appears that the singularity of $S_{W}$ occurs frequently, particularly
in the case of very high-dimensional space or in the case of under-sampled
problems. In the literature, different solutions~\cite{Friedman89,Fukunaga90,Hastie95,Howland03,Jin01}
are proposed to deal with such a problem in a supervised classification
framework. In addition, since clustering approaches are sensitive
to high-dimensional and noisy data, recent works~\cite{Ding07,Torre06,Xu04,Ye07}
focused on combining low dimensional discriminative subspace with
one of the most used clustering algorithm: k-means. However, these
approaches do not really compute the discriminant subspace and are
not interested in the visualization and the understanding of the data.

\section{Model-based clustering in a discriminative subspace\label{sec:Model-DLM}}

This section introduces a mixture model, called the discriminative
latent mixture model, which aims to find both a parsimonious and discriminative
fit for the data in order to generate a clustering and a visualization
of the data. The modeling proposed in this section is mainly based
on two key ideas: firstly, actual data are assumed to live in a latent
subspace with an intrinsic dimension lower than the dimension of the
observed data and, secondly, a subspace of $K-1$ dimensions is theoretically
sufficient to discriminate $K$ groups.

\subsection{The discriminative latent mixture model}

Let $\{y_{1},\dots,y_{n}\}\in\mathbb{R}^{p}$ denote a dataset of
$n$ observations that one wants to cluster into $K$ homogeneous
groups, \emph{i.e.} adjoin to each observation $y_{j}$ a value $z_{j}\in\{1,\dots,K\}$
where $z_{i}=k$ indicates that the observation $y_{i}$ belongs to
the $k^{th}$ group. On the one hand, let us assume that $\{y_{1},\dots,y_{n}\}$
are independent observed realizations of a random vector $Y\in\mathbb{R}^{p}$
and that $\{z_{1},\dots,z_{n}\}$ are also independent realizations
of a random vector $Z\in\{1,\dots,K\}$. On the other hand, let $\mathbb{E}\subset\mathbb{R}^{p}$
denote a latent space assumed to be the most discriminative subspace
of dimension $d\leq K-1$ such that $\mathbf{0}\in\mathbb{E}$ and
where $d$ is strictly lower than the dimension $p$ of the observed
space. Moreover, let $\{x_{1},\dots,x_{n}\}\in\mathbb{E}$ denote
the actual data, described in the latent space~$\mathbb{E}$ of dimension
$d$, which are in addition presumed to be independent unobserved
realizations of a random vector $X\in\mathbb{E}$. Finally, for each
group, the observed variable $Y\in\mathbb{R}^{p}$ and the latent
variable $X\in\mathbb{E}$ are assumed to be linked through a linear
transformation: \begin{equation}
Y=UX+\varepsilon,\label{e:model}\end{equation}
 where $d<p$, $U$ is the $p\times d$ orthogonal matrix common to
the $K$ groups, such as $U^{t}U=I_{d}$, and $\varepsilon\in\mathbb{R}^{p}$,
conditionally to $Z$, is a centered Gaussian noise term with covariance
matrix $\Psi_{k}$, for $k=1,...,K$: \[
\varepsilon_{|Z=k}\sim\mathcal{N}(\mathbf{0},\Psi_{k}).\]
Following the classical framework of model-based clustering, each
group is in addition assumed to be distributed according to a Gaussian
density function within the latent space $\mathbb{E}$. Hence, the
random vector $X\in\mathbb{E}$ has the following conditional density
function: \[
X|Z=k\sim\mathcal{N}(\mu_{k},\Sigma_{k}),\]
 where $\mu_{k}\in\mathbb{R}^{d}$ and $\Sigma_{k}\in\mathbb{R}^{d\times d}$
are respectively the mean and the covariance matrix of the $k$th
group. Conditionally to $X$ and $Z$, the random vector $Y\in\mathbb{R}^{d}$
has the following conditional distribution: \[
Y|X,Z=k\sim\mathcal{N}(UX,\Psi_{k}),\]
 and its marginal distribution is therefore a mixture of Gaussians:
\[
f(y)=\sum_{k=1}^{K}\pi_{k}\phi(y;m_{k},S_{k}),\]
 where $\pi_{k}$ is the mixture proportion of the $k$th group and:
\begin{align*}
m_{k} & =U\mu_{k},\\
S_{k} & =U\Sigma_{k}U^{t}+\Psi_{k},\end{align*}
 are respectively the mean and the covariance matrix of the $k$th
group in the observation space. Let us also define $W=[U,V]$ a $p\times p$
matrix which satisfies $W^{t}W=WW^{t}=I_{p}$ and for which the $p\times(p-d)$
matrix $V$, is the orthonormal complement of $U$ defined above.
We finally assume that the noise covariance matrix $\Psi_{k}$ satisfies
the conditions $V\Psi_{k}V^{t}=\beta_{k}I_{d-p}$ and $U\Psi_{k}U^{t}=0_{d}$,
such that $\Delta_{k}=W^{t}S_{k}W$ has the following form: $$\Delta_k=\left(  \begin{array}{c@{}c} \begin{array}{|ccc|}\hline ~~ & ~~ & ~~ \\  & \Sigma_k &  \\  & & \\ \hline \end{array} & \mathbf{0}\\ \mathbf{0} &  \begin{array}{|cccc|}\hline \beta_k & & & 0\\ & \ddots & &\\  & & \ddots &\\ 0 & & & \beta_k\\ \hline \end{array} \end{array}\right)  \begin{array}{cc} \left.\begin{array}{c} \\\\\\\end{array}\right\}  & d \leq K-1\vspace{1.5ex}\\ \left.\begin{array}{c} \\\\\\\\\end{array}\right\}  & (p-d)\end{array}$$%
\begin{figure}[t]
\centering{}\psfrag{X}{\hspace{-1ex}$X$}
\psfrag{Y}{\hspace{-1ex}$Y$}
\psfrag{Z}{\hspace{-1ex}$Z$}
\psfrag{pi}{\hspace{-15ex}$\pi=\{\pi_1,...,\pi_K\}$}
\psfrag{mu}{\hspace{0ex}$\mu_k\in\mathbb{E}$}
\psfrag{sig}{\hspace{0ex}$\Sigma_k$}
\psfrag{V}{\hspace{0ex}$W=[U,V]$}
\psfrag{e}{\hspace{-1ex}$\varepsilon$}
\psfrag{psi}{\hspace{-1ex}$\Psi_k$}\includegraphics[width=0.5\columnwidth]{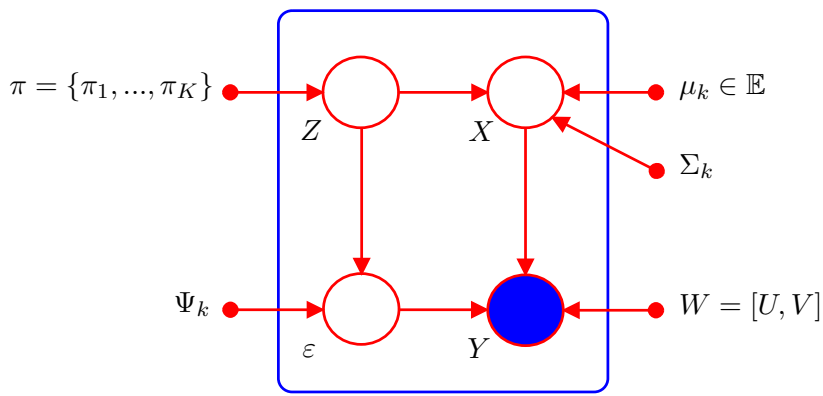}
\caption{\label{Fig:schema1} Graphical summary of the $\mathrm{DLM}_{[\Sigma_{k}\beta_{k}]}$
model}

\end{figure}
\hspace{-1ex}This model, called the discriminative latent mixture~(DLM)
model and referred to by $\mathrm{DLM}_{[\Sigma_{k}\beta_{k}]}$ in
the sequel, is summarized by Figure~\ref{Fig:schema1}. The $\mathrm{DLM}_{[\Sigma_{k}\beta_{k}]}$
model is therefore parametrized by the parameters $\pi_{k}$, $\mu_{k}$,
$U$, $\Sigma_{k}$ and $\beta_{k}$, for $k=1,...,K$ and $j=1,...,d$.
On the one hand, the mixture proportions $\pi_{1},...,\pi_{K}$ and
the means $\mu_{1},...,\mu_{K}$ parametrize in a classical way the
prior probability and the average latent position of each group respectively.
On the other hand, $U$ defines the latent subspace $\mathbb{E}$
by parametrizing its orientation according to the basis of the original
space. Finally, $\Sigma_{k}$ parametrize the variance of the $k$th
group within the latent subspace $\mathbb{E}$ whereas $\beta_{k}$
parametrizes the variance of this group outside~$\mathbb{E}$. With
these notations and from a practical point of view, one can say that
the variance of the actual data is therefore modeled by $\Sigma_{k}$
and the variance of the noise is modeled by $\beta_{k}$.

\subsection{The submodels of the $\mathrm{DLM}_{[\Sigma_{k}\beta_{k}]}$ model}

Starting with the $\mathrm{DLM}_{[\Sigma_{k}\beta_{k}]}$ model presented
in the previous paragraph, several submodels can be generated by applying
constraints on parameters of the matrix $\Delta_{k}$. For instance,
the covariance matrices $\Sigma_{1},\dots,\Sigma_{K}$ in the latent
space can be assumed to be common across groups and this submodel
will be referred to by $\mathrm{DLM}_{[\Sigma\beta_{k}]}$. Similarly,
in each group, $\Sigma_{k}$ can be assumed to be diagonal, \emph{i.e.}
$\Sigma_{k}=\mathrm{{diag}}(\alpha_{k1},\dots,\alpha_{kd})$. This
submodel will be referred to by $\mathrm{DLM}_{[\alpha_{kj}\beta_{k}]}$.
In the same manner, the $p-d$ last values of $\Delta_{k}$ can be
assumed to be common for the $k$ classes, \emph{i.e.} $\beta_{k}=\beta$,
$\forall k=1,...,K$, meaning that the variance outside the discriminant
subspace is common to all groups. This assumption can be viewed as
modeling the noise variance with a unique parameter which seems natural
for data obtained in a common acquisition process. Following the notation
system introduces above, this submodel will be referred to by $\mathrm{DLM}_{[\alpha_{kj}\beta]}$.
The variance within the latent subspace $\mathbb{E}$ can also be
assumed to be isotropic for each group and the associated submodel
is $\mathrm{DLM}_{[\alpha_{k}\beta_{k}]}$. In this case, the variance
of the data is assumed to be isotropic both within~$\mathbb{E}$
and outside~$\mathbb{E}$. Similarly, it is possible to constrain
the previous model to have the parameters $\beta_{k}$ common between
classes and this gives rise to the model~$\mathrm{DLM}_{[\alpha_{k}\beta]}$.
Finally, the variance within the subspace~$\mathbb{E}$ can be assumed
to be independent from the mixture component and this corresponds
to the models~$\mathrm{DLM}_{[\alpha_{j}\beta_{k}]}$, $\mathrm{DLM}_{[\alpha_{j}\beta]}$,
$\mathrm{DLM}_{[\alpha\beta_{k}]}$ and $\mathrm{DLM}_{[\alpha\beta]}$.
We therefore enumerate $12$ different DLM models and an overview
of them is proposed in Table~\ref{Tab:models}.%
\begin{table}
\begin{centering}
\doublespacing{\begin{tabular}{lcc}
Model  & Nb. of parameters  & $K=4$ and $p=100$\tabularnewline
\hline 
$\mathrm{DLM}_{[\Sigma_{k}\beta_{k}]}$ & $(K-1)+K(K-1)+(K-1)(p-K/2)+K^{2}(K-1)/2+K$ & 337\tabularnewline
$\mathrm{DLM}_{[\Sigma_{k}\beta]}$ & $(K-1)+K(K-1)+(K-1)(p-K/2)+K^{2}(K-1)/2+1$ & 334\tabularnewline
$\mathrm{DLM}_{[\Sigma\beta_{k}]}$ & $(K-1)+K(K-1)+(K-1)(p-K/2)+K(K-1)/2+K$ & 319\tabularnewline
$\mathrm{DLM}_{[\Sigma\beta]}$ & $(K-1)+K(K-1)+(K-1)(p-K/2)+K(K-1)/2+1$ & 316\tabularnewline
$\mathrm{DLM}_{[\alpha_{kj}\beta_{k}]}$ & $(K-1)+K(K-1)+(K-1)(p-K/2)+K^{2}$  & 325\tabularnewline
$\mathrm{DLM}_{[\alpha_{kj}\beta]}$ & $(K-1)+K(K-1)+(K-1)(p-K/2)+K(K-1)+1$  & 322\tabularnewline
$\mathrm{DLM}_{[\alpha_{k}\beta_{k}]}$ & $(K-1)+K(K-1)+(K-1)(p-K/2)+2K$  & 317\tabularnewline
$\mathrm{DLM}_{[\alpha_{k}\beta]}$ & $(K-1)+K(K-1)+(K-1)(p-K/2)+K+1$  & 314\tabularnewline
$\mathrm{DLM}_{[\alpha_{j}\beta_{k}]}$ & $(K-1)+K(K-1)+(K-1)(p-K/2)+(K-1)+K$  & 316\tabularnewline
$\mathrm{DLM}_{[\alpha_{j}\beta]}$ & $(K-1)+K(K-1)+(K-1)(p-K/2)+(K-1)+1$  & 313\tabularnewline
$\mathrm{DLM}_{[\alpha\beta_{k}]}$ & $(K-1)+K(K-1)+(K-1)(p-K/2)+K+1$  & 314\tabularnewline
$\mathrm{DLM}_{[\alpha\beta]}$ & $(K-1)+K(K-1)+(K-1)(p-K/2)+2$ & 311\tabularnewline
\hline 
Full-GMM  & $(K-1)+Kp+Kp(p+1)/2$  & 20603 \tabularnewline
Com-GMM  & $(K-1)+Kp+p(p+1)/2$  & 5453\tabularnewline
Mixt-PPCA  & $(K-1)+Kp+K(d(p-(d+1)/2)+d+1)+1$ & 1198 ($d=3$)\tabularnewline
Diag-GMM & $(K-1)+Kp+Kp$  & 803 \tabularnewline
Sphe-GMM & $(K-1)+Kp+K$  & 407 \tabularnewline
\hline
\end{tabular}}
\par\end{centering}

\caption{\label{Tab:models} Number of free parameters to estimate when $d=K-1$
for the DLM models and some classical models (see text for details).}

\end{table}
 The table also gives the maximum number of free parameters to estimate
(case of $d=K-1$) according to $K$ and $p$ for the 12 DLM models
and for some classical models. The Full-GMM model refers to the classical
Gaussian mixture model with full covariance matrices, the Com-GMM
model refers to the Gaussian mixture model for which the covariance
matrices are assumed to be equal to a common covariance matrix ($S_{k}=S$,
$\forall k$), Diag-GMM refers to the Gaussian mixture model for which
$S_{k}=\mathrm{diag}(s_{k1}^{2},...,s_{kp}^{2})$ with $s_{k}^{2}\in\mathbb{R}^{p}$
and Sphe-GMM refers to the Gaussian mixture model for which $S_{k}=s_{k}^{2}I_{p}$
with $s_{k}^{2}\in\mathbb{R}$. Finally, Mixt-PPCA denotes the subspace
clustering model proposed by Tipping and Bishop in~\cite{Tipping99}.
In addition to the number of free parameters to estimate, Table~\ref{Tab:models}
gives this number for specific values of $K$ and $p$ in the right
column. The number of free parameters to estimate given in the central
column can be decomposed in the number of parameters to estimate for
the proportions ($K-1$), for the means ($Kp$) and for the covariance
matrices (last terms). Among the classical models, the Full-GMM model
is a highly parametrized model and requires the estimation of 20603
parameters when $K=4$ and $p=100$. Conversely, the Diag-GMM and
Sphe-GMM model are very parsimonious models since they respectively
require the estimation of only 803 and 407 parameters when $K=4$
and $p=100$. The Com-GMM and Mixt-PPCA models appear to both have
an intermediate complexity. However, the Mixt-PPCA model is a less
constrained model compared to the Diag-GMM model and should be preferred
for clustering high-dimensional data. Finally, the DLM models turn
out to have a low complexity whereas their modeling capacity is comparable
to the one of the Mixt-PPCA model. In addition, the complexity of
the DLM models depends only on $K$ and $p$ whereas the Mixt-PPCA
model depends from an hyper-parameter $d$.

\subsection{Links with existing models}

At this point, some links can be established with models existing
in the clustering literature. The closest models have been proposed
in~\cite{Baek2009}, \cite{Bouveyron07} and~\cite{McNicholas2008}
and are all derived from the mixture of factor analyzer (MFA) model~\cite{McLachlan2003,Rubin82}.
First, in~\cite{Bouveyron07}, the authors proposed a family of 28
parsimonious and flexible Gaussian models ranging from a very general
model, referred to as $[a_{kj}b_{k}Q_{k}d_{k}]$, to very simple models.
Compared to the standard MFA model, these parsimonious models assume
that the noise variance is isotropic. In particular, this work can
be viewed as an extension of the mixture of principal component analyzer
(Mixt-PPCA) model~\cite{Tipping99}. Among this family of parsimonious
models, 14 models assume that the orientation of the group-specific
subspaces is common (common $Q_{k}$). The following year, McNicholas
and Murphy~\cite{McNicholas2008} proposed as well a family of 8
parsimonious Gaussian models by extending the MFA model by constraining
the loading and error variance matrices across groups. In this work,
the noise variance can be isotropic or not. Let us remark that the
two families of parsimonious Gaussian models share some models: for
instance, the model UUC of~\cite{McNicholas2008} corresponds to
the model~$[a_{kj}b_{k}Q_{k}d]$ of~\cite{Bouveyron07}. Among the
8 parsimonious models presented in~\cite{McNicholas2008}, 4 models
have the loading matrices constrained across the groups. More recently,
Beak \emph{et al.}~\cite{Baek2009} proposed as well a MFA model
with a common loading matrix. In this case, the noise variance is
not constrained. Despite their differences, all these parsimonious
Gaussian models share the assumption that the group subspaces have
a common orientation and are therefore close to the DLM models presented
in this work. However, these models with common loadings choose the
orientation such as the variance of the projected data is maximum
whereas the DLM models choose the latent subspace orientation such
as it best discriminates the groups. This specific feature of the
DLM models should therefore improve in most cases both the clustering
and the visualization of the results. In particular, the DLM models
should be able to better model situations where the axes carrying
the greatest variance are not parallel to the discriminative axes
than the other approaches (Figure~10.1 of~\cite{Fukunaga90} illustrates
such a situation).

\section{Parameter estimation: the Fisher-EM algorithm\label{sec:Parameter-Estimation}}

Since this work focuses on the clustering of unlabeled data, this
section introduces an estimation procedure which adapts the traditional
EM algorithm for estimating the parameters of DLM models presented
in the previous section. Due to the nature of the models described
above, the Fisher-EM algorithm alternates between three steps: 
\begin{itemize}
\item an E step in which posterior probabilities that observations belong
to the $K$ groups are computed, 
\item a F step which estimates the orientation matrix $U$ of the discriminative
latent space conditionally to the posterior probabilities, 
\item a M step in which parameters of the mixture model are estimated in
the latent subspace by maximizing the conditional expectation of the
complete likelihood. 
\end{itemize}
This estimation procedure relative to the DLM models is called hereafter
the Fisher-EM algorithm. We chose to name this estimation procedure
after Sir R.~A.~Fisher since the key idea of the F step comes from
his famous work on discrimination. The remainder of this section details
the simple form of this procedure. Let us however notice that the
Fisher-EM algorithm can be also used in combination with the stochastic~\cite{Celeux85}
and classification versions~\cite{Celeux92} of the EM algorithm.

\subsection{The E step}

This step aims to compute, at iteration $(q)$, the expectation of
the complete log-likelihood conditionally to the current value of
the parameter $\theta^{(q-1)}$, which, in practice, reduces to the
computation of $t_{ik}^{(q)}=E[z_{ik}|y_{i},\theta^{(q-1)}]$ where
$z_{ik}=1$ if $y_{i}$ comes from the $k$th component and $z_{ik}=0$
otherwise. Let us also recall that $t_{ik}^{(q)}$ is as well the
posterior probability that the observation $y_{i}$ belongs to the
$k^{th}$ component of the mixture. The following proposition provides
the explicit form of $t_{ik}^{(q)}$, for $i=1,...,n$, $k=1,...,K$,
in the case of the model $\mathrm{DLM}_{[\Sigma_{k}\beta_{k}]}$.
Demonstration of this result is detailed in Appendix~\ref{a:costfunc}.

\begin{prop} \label{prop:Estep}With the assumptions of the model
$\mathrm{DLM}_{[\Sigma_{k}\beta_{k}]}$, the posterior probabilities~$t_{ik}^{(q)}$,
$i=1,...,n$, $k=1,...,K$, can be expressed as : \[
t_{ik}^{(q)}=\frac{1}{\sum_{l=1}^{K}\exp\left(\frac{1}{2}(\Gamma_{k}^{(q-1)}(y)-\Gamma_{l}^{(q-1)}(y))\right)},\]
 with: \begin{equation}
\begin{split}\Gamma_{k}^{(q-1)}(y_{i})= & \quad||P(y_{i}-m_{k}^{(q-1)})||_{\mathcal{D}_{k}}^{2}+\frac{1}{\beta_{k}^{(q-1)}}||(y_{i}-m_{k}^{(q-1)})-P(y_{i}-m_{k}^{(q-1)})||^{2}\\
 & +\log\left(\left|\Sigma_{k}^{(q-1)}\right|\right)+(p-d)\log(\beta_{k}^{(q-1)})-2\log(\pi_{k}^{(q-1)})+\gamma,\end{split}
\label{Eq_Estep}\end{equation}
 where $||.||_{\mathcal{D}_{k}}^{2}$ is a norm on the latent space
$\mathbb{E}$ defined by $||y||_{\mathcal{D}_{k}}^{2}=y^{t}\mathcal{D}_{k}y$,
$\mathcal{D}_{k}=\tilde{W}\Delta_{k}^{-1}\tilde{W}^{t}$, $\tilde{W}$
is a $p\times p$ matrix containing the $d$ vectors of $U^{(q-1)}$
completed by zeros such as $\tilde{W}=[U^{(q-1)},0_{p-d}]$, $P$
is the projection operator on the latent space $\mathbb{E}$, \emph{i.e.}
$P(y)=U^{(q-1)}U^{(q-1)t}y$, and $\gamma=p\log(2\pi)$ is a constant
term. \end{prop}

Besides its computational interest, Proposition~\ref{prop:Estep}
provides as well a comprehensive interpretation of the cost function
$\Gamma_{k}$ which mainly governs the computation of $t_{ik}$. Indeed,
it appears that $\Gamma_{k}$ mainly depends on two distances: the
distance between the projections on the discriminant subspace $\mathbb{E}$
of the observation $y_{i}$ and the mean $m_{k}$ on the one hand,
and, the distance between the projections on the complementary subspace
$\mathbb{E}^{\perp}$ of $y_{i}$ and $m_{k}$ on the other hand.
Remark that the latter distance can be reformulated in order to avoid
the use of the projection on $\mathbb{E}^{\perp}$. Indeed, as Figure~\ref{Fig:projection}
illustrates, this distance can be re-expressed according projections
on $\mathbb{E}$. Therefore, the posterior probability $t_{ik}=P(z_{ik}=1|y_{i})$
will be close to $1$ if both the distances are small which seems
quite natural. Obviously, these distances are also balanced by the
variances in $\mathbb{E}$ and $\mathbb{E}^{\perp}$ and by the mixture
proportions. Furthermore, the fact that the E step does not require
the use of the projection on the complementary subspace $\mathbb{E}^{\perp}$
is, from a computational point of view, very important because it
will provide the stability of the algorithm and will allow its use
when $n\ll p$ (\emph{cf.} paragraph~\ref{par_n<<p}).

\begin{figure}[t]

\centering{}\psfrag{mm}{$m_{1}$} \psfrag{m}{$m_{2}$} \psfrag{M}{$\mu_{1}$}
\psfrag{x}{\vspace{-1ex}
$x$} \psfrag{y}{\hspace{0ex}$y$} \psfrag{E}{\hspace{-1ex}$\mathbb{E}$}
\psfrag{F}{$\mathbb{E}^{\perp}$} \includegraphics[width=0.6\columnwidth]{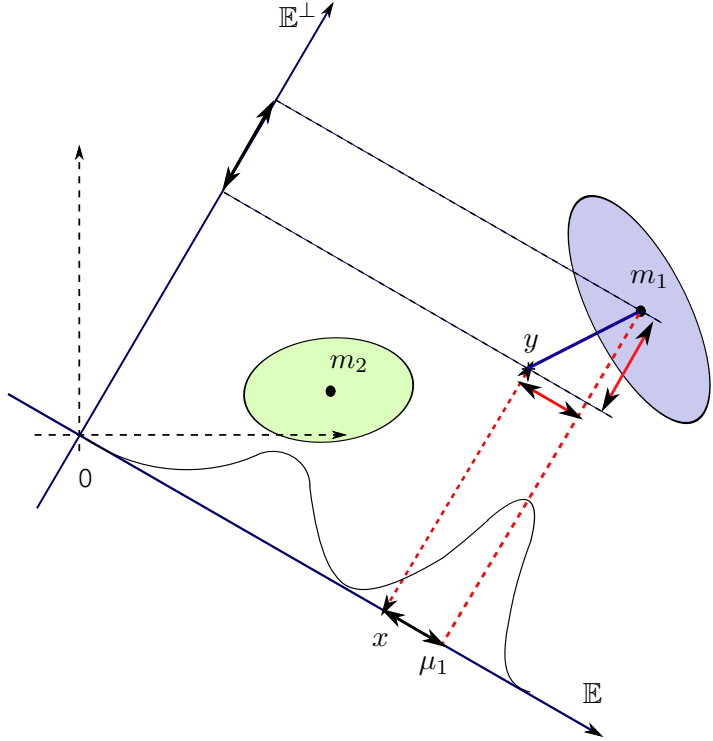}
\caption{\label{Fig:projection} Two groups and their $1$-dimensional discriminative
subspace.}

\end{figure}

\subsection{The F step}

This step aims to determinate, at iteration $(q)$, the discriminative
latent subspace of dimension $d\leq K-1$ in which the $K$ groups
are best separated. Naturally, the estimation of this latent subspace
has to be done conditionally to the current values of posterior probabilities
$t_{ik}^{(q)}$ which indicates the current soft partition of the
data. Estimating the discriminative latent subspace $\mathbb{E}^{(q)}$
reduces to the computation of $d$ discriminative axes. Following
the original idea of Fisher~\cite{Fisher36}, the $d$ axes which
best discriminate the $K$ groups are those which maximize the traditional
criterion $J(U)=tr((U^{t}S_{W}U)^{-1}U^{t}S_{B}U)$. However, the
traditional criterion $J(U)$ assume that the data are complete (supervised
classification framework). Unfortunately, the situation of interest
here is that of unsupervised classification and the matrices $S_{B}$
and $S_{W}$ have therefore to be defined conditionally to the current
soft partition. Furthermore, the DLM models assume that the discriminative
latent subspace must have an orthonormal basis and, sadly, the traditional
Fisher's approach provides non-orthogonal discriminative axes.

To overcome both problems, this paragraph proposes a procedure which
keeps the key idea of Fisher while providing orthonormal discriminative
axes conditionally to the current soft partition of the data. The
procedure follows the concept of the orthonormal discriminant vector
(ODV) method introduced by~\cite{Foley75} in the supervised case
and then extended by~\cite{Guo01,Hamamoto91,Liu92,Ye05}, which sequentially
selects the most discriminative features in maximizing the Fisher
criterion subject to the orthogonality of features. First, it is necessary
to introduce the soft between-covariance matrix $S_{B}^{(q)}$ and
the soft within-covariance matrix $S_{W}^{(q)}$. The soft between-covariance
matrix $S_{B}^{(q)}$ is defined conditionally to the posterior probabilities
$t_{ik}^{(q)}$, obtained in the E step, as follows: \[
S_{B}^{(q)}=\frac{1}{n}\sum_{k=1}^{K}n_{k}^{(q)}(\hat{m}_{k}^{(q)}-\bar{y})(\hat{m}_{k}^{(q)}-\bar{y})^{t},\]
 where $n_{k}^{(q)}=\sum_{i=1}^{n}t_{ik}^{(q)}$, $\hat{m}_{k}^{(q)}=\frac{1}{n}\sum_{i=1}^{n}t_{ik}^{(q)}y_{i}$
is the soft mean of the $k$th group at iteration $q$ and $\bar{y}=\frac{1}{n}\sum_{i=1}^{n}y_{i}$
is the empirical mean of the whole dataset. Since the relation $S=S_{W}^{(q)}+S_{B}^{(q)}$
holds in this context as well, it is preferable from a computational
point of view to use the covariance matrix $S=\frac{1}{n}\sum_{i=1}^{n}(y_{i}-\bar{y})(y_{i}-\bar{y})^{t}$
of the whole dataset in the maximization problem instead of $S_{W}^{(q)}$
since $S$ remains fixed over the iteration. The F step of the Fisher-EM
therefore aims to solve the following optimization problem: \[
\left\{ \begin{aligned}\max_{U}\quad & \quad\mathrm{trace}\left((U^{t}SU)^{-1}U^{t}S_{B}^{(q)}U\right),\\
\quad\text{wrt}\quad & \quad u_{j}^{t}u_{l}=0,\quad\forall j\neq l\in\{1,\dots,d\},\end{aligned}
\right.\]
where $u_{j}$ is the $j$th column vector of $U$. Following the
ODV procedure, the $d$ axes solution of this optimization problem
are iteratively constructed by, first, computing an orthogonal complementary
subspace to the current set of discriminative axes and, then, maximizing
the Fisher criterion in this orthogonal subspace by solving the associated
generalized eigenvalue problem. To initialize this iterative procedure,
the first vector of $U$ is therefore the eigenvector associated with
the largest eigenvalue of the matrix $S^{-1}S_{B}^{(q)}$. Then, assuming
that the $r-1$ first orthonormal discriminative axes $\{u_{1},\dots,u_{r-1}\}$,
which span the space $\mathcal{B}_{r-1}$, have been computed, the
$r^{th}$ discriminative axis has to lie in the subspace $\mathcal{B}_{r-1}^{\perp}$
orthogonal to the space $\mathcal{B}_{r-1}$. The Gram-Schmidt orthonormalization
procedure allows to find a basis $V^{r}=\{v_{r},v_{r+1},...,v_{d}\}$
for the orthogonal subspace $\mathcal{B}_{r-1}^{\perp}$ such that:
\[
v_{l}=\alpha_{l}(I_{\ell-1}-\sum_{j=1}^{\ell-1}v_{j}v_{j}^{t})\psi_{l},\qquad\ell=r,\dots,p\]
 where $v_{j}=u_{j}$ for $j=1,...,r-1$, $\alpha_{\ell}$ is normalization
constant such that $||u_{\ell}||=1$ and $\psi_{\ell}$ is a vector
linearly independent of $u_{j}~\forall j\in\{1,\dots,\ell-1\}$. Then,
the $r$th discriminative axis is given by: \[
u_{r}=\frac{P_{r-1}u_{r}^{max}}{||u_{r}^{max}||},\]
 where $P_{r-1}$ is the projector on $\mathcal{B}_{r-1}$, $u_{r}^{max}$
is the eigenvector associated with the largest eigenvalue of the matrix
$S_{r}^{-1}S_{Br}^{(q)}$ with: \begin{align*}
S_{r}= & \quad V^{r^{{\scriptstyle t}}}SV^{r},\\
S_{Br}^{(q)}= & \quad V^{r^{{\textstyle {\scriptstyle t}}}}S_{B}^{(q)}V^{r},\end{align*}
 \emph{i.e.} $S_{r}$ and $S_{Br}^{(q)}$ are respectively the covariance
and soft between-covariance matrices of the data projected into the
orthogonal subspace $\mathcal{B}_{r-1}^{\perp}$. This iterative procedure
stops when the $d$ orthonormal discriminative axes $u_{j}$ are computed.

\subsection{The M step}

This third step estimates the model parameters by maximizing the conditional
expectation of the complete likelihood. The following proposition
provides the expression of the conditional expectation of the complete
log-likelihood in the case of the $\mathrm{DLM}_{[\Sigma_{k}\beta_{k}]}$
model. A proof of this result is provided in Appendix~\ref{a:estimation}

\begin{prop}\label{prop:Likelihood}In the case of the model $\mathrm{DLM}_{[\Sigma_{k}\beta_{k}]}$,
the conditional expectation of complete log-likelihood $Q(y_{1},\dots,y_{n},\theta)$
has the following expression: \begin{equation}
\begin{split}Q(y_{1},\dots,y_{n},\theta)= & -\frac{1}{2}\sum_{k=1}^{K}n_{k}\Bigl[-2\log(\pi_{k})+\mathrm{trace}(\Sigma_{k}^{-1}U^{t}C_{k}U)+\log(\left|\Sigma_{k}\right|)\\
 & +(p-d)\log(\beta_{k})+\frac{1}{\beta_{k}}\left(\mathrm{trace}(C_{k})-\sum_{j=1}^{d}u_{j}^{t}C_{k}u_{j}\right)+\gamma\Bigr].\end{split}
\label{e:loglik2}\end{equation}
 where $C_{k}$ is the empirical covariance matrix of the $k^{th}$
group, $u_{j}$ is the $j$th column vector of~$U$, $n_{k}=\sum_{i=1}^{n}t_{ik}$
and $\gamma=p\log(2\pi)$ is a constant term. \end{prop}

\noindent At iteration $q$, the maximization of $Q$ conduces to
an estimation of the mixture proportions $\pi_{k}$ and the means
$\mu_{k}$ for the $K$ components by their empirical counterparts:
\begin{align*}
\hat{\pi}_{k}^{(q)} & =\frac{n_{k}}{n},\\
\hat{\mu}_{k}^{(q)} & =\frac{1}{n_{k}}\sum_{i=1}^{n}t_{ik}^{(q)}U^{(q)t}y_{i},\end{align*}
where $n_{k}=\sum_{i=1}^{n}t_{ik}^{(q)}$ and $U^{(q)}$ contains,
as columns vectors, the $d$ discriminative axes $u_{j}^{(q)}$, $j=1,...,d$,
provided by the F step at iteration $q$. The following proposition
provides estimates for the remaining parameters for the 12 DLM models
which have to be updated at each iteration of the FEM procedure. Proofs
of the following results are given in Appendix~\ref{a:estimation}.

\begin{prop} \label{prop:MLestimates}At iteration $q$, the estimates
for variance parameters of the 12 DLM models are: 
\begin{itemize}
\item Model $\mathrm{DLM}_{[\Sigma_{k}\beta_{k}]}$:\begin{align}
\hat{\Sigma}_{k}^{(q)}=U^{(q)t}C_{k}^{(q)}U^{(q)}, & \quad\hat{\beta}_{k}^{(q)}=\frac{\mathrm{trace}(C_{k}^{(q)})-\sum_{j=1}^{d}u_{j}^{(q)t}C_{k}^{(q)}u_{j}^{(q)}}{p-d},\label{Eq_Mstep_debut-1}\end{align}

\item Model $\mathrm{DLM}_{[\Sigma_{k}\beta]}$:\begin{align}
\hat{\Sigma}_{k}^{(q)}=U^{(q)t}C_{k}^{(q)}U^{(q)}, & \quad\hat{\beta}^{(q)}=\frac{\mathrm{trace}(C^{(q)})-\sum_{j=1}^{d}u_{j}^{(q)t}C^{(q)}u_{j}^{(q)}}{p-d},\label{Eq_Mstep_debut-1}\end{align}

\item Model $\mathrm{DLM}_{[\Sigma\beta_{k}]}$:\begin{align}
\hat{\Sigma}^{(q)}=U^{(q)t}C^{(q)}U^{(q)}, & \quad\hat{\beta}_{k}^{(q)}=\frac{\mathrm{trace}(C_{k}^{(q)})-\sum_{j=1}^{d}u_{j}^{(q)t}C_{k}^{(q)}u_{j}^{(q)}}{p-d},\label{Eq_Mstep_debut-1}\end{align}

\item Model $\mathrm{DLM}_{[\Sigma\beta]}$:\begin{align}
\hat{\Sigma}^{(q)}=U^{(q)t}C^{(q)}U^{(q)}, & \quad\hat{\beta}^{(q)}=\frac{\mathrm{trace}(C^{(q)})-\sum_{j=1}^{d}u_{j}^{(q)t}C^{(q)}u_{j}^{(q)}}{p-d},\label{Eq_Mstep_debut-1}\end{align}

\item Model $\mathrm{DLM}_{[\alpha_{kj}\beta_{k}]}$:\begin{align}
\hat{\alpha}_{kj}^{(q)}=u_{j}^{(q)t}C_{k}^{(q)}u_{j}^{(q)}, & \quad\hat{\beta}_{k}^{(q)}=\frac{\mathrm{trace}(C_{k}^{(q)})-\sum_{j=1}^{d}u_{j}^{(q)t}C_{k}^{(q)}u_{j}^{(q)}}{p-d},\label{Eq_Mstep_debut-1}\end{align}

\item Model $\mathrm{DLM}_{[\alpha_{kj}\beta]}$:\begin{align}
\hat{\alpha}_{kj}^{(q)}=u_{j}^{(q)t}C_{k}^{(q)}u_{j}^{(q)}, & \quad\hat{\beta}^{(q)}=\frac{\mathrm{trace}(C^{(q)})-\sum_{j=1}^{d}u_{j}^{(q)t}C^{(q)}u_{j}^{(q)}}{p-d},\label{Eq_Mstep_debut-1}\end{align}

\item Model $\mathrm{DLM}_{[\alpha_{k}\beta_{k}]}$:\begin{align}
\hat{\alpha}_{k}^{(q)}=\frac{1}{d}\sum_{j=1}^{d}u_{j}^{(q)t}C_{k}^{(q)}u_{j}^{(q)}, & \quad\hat{\beta}_{k}^{(q)}=\frac{\mathrm{trace}(C_{k}^{(q)})-\sum_{j=1}^{d}u_{j}^{(q)t}C_{k}^{(q)}u_{j}^{(q)}}{p-d},\label{Eq_Mstep_debut-1}\end{align}

\item Model $\mathrm{DLM}_{[\alpha_{k}\beta]}$:\begin{align}
\hat{\alpha}_{k}^{(q)}=\frac{1}{d}\sum_{j=1}^{d}u_{j}^{(q)t}C_{k}^{(q)}u_{j}^{(q)}, & \quad\hat{\beta}^{(q)}=\frac{\mathrm{trace}(C^{(q)})-\sum_{j=1}^{d}u_{j}^{(q)t}C^{(q)}u_{j}^{(q)}}{p-d},\label{Eq_Mstep_debut-1}\end{align}

\item Model $\mathrm{DLM}_{[\alpha_{j}\beta_{k}]}$:\begin{align}
\hat{\alpha}_{j}^{(q)}=u_{j}^{(q)t}C^{(q)}u_{j}^{(q)}, & \quad\hat{\beta}_{k}^{(q)}=\frac{\mathrm{trace}(C_{k}^{(q)})-\sum_{j=1}^{d}u_{j}^{(q)t}C_{k}^{(q)}u_{j}^{(q)}}{p-d},\label{Eq_Mstep_debut-1}\end{align}

\item Model $\mathrm{DLM}_{[\alpha_{j}\beta]}$:\begin{align}
\hat{\alpha}_{j}^{(q)}=u_{j}^{(q)t}C^{(q)}u_{j}^{(q)}, & \quad\hat{\beta}^{(q)}=\frac{\mathrm{trace}(C^{(q)})-\sum_{j=1}^{d}u_{j}^{(q)t}C^{(q)}u_{j}^{(q)}}{p-d},\label{Eq_Mstep_debut-1}\end{align}

\item Model $\mathrm{DLM}_{[\alpha\beta_{k}]}$:\begin{align}
\hat{\alpha}^{(q)}=\frac{1}{d}\sum_{j=1}^{d}u_{j}^{(q)t}C^{(q)}u_{j}^{(q)}, & \quad\hat{\beta}_{k}^{(q)}=\frac{\mathrm{trace}(C_{k}^{(q)})-\sum_{j=1}^{d}u_{j}^{(q)t}C_{k}^{(q)}u_{j}^{(q)}}{p-d},\label{Eq_Mstep_debut-1}\end{align}

\item Model $\mathrm{DLM}_{[\alpha\beta]}$:\begin{align}
\hat{\alpha}^{(q)}=\frac{1}{d}\sum_{j=1}^{d}u_{j}^{(q)t}C^{(q)}u_{j}^{(q)}, & \quad\hat{\beta}^{(q)}=\frac{\mathrm{trace}(C^{(q)})-\sum_{j=1}^{d}u_{j}^{(q)t}C^{(q)}u_{j}^{(q)}}{p-d},\label{DUPLICATION : Eq_Mstep_fin}\end{align}

\end{itemize}
where the vectors $u_{j}^{(q)}$ are the discriminative axes provided
by the F step at iteration $q$, $C_{k}^{(q)}=\frac{1}{n_{k}^{(q)}}\sum_{i=1}^{n}t_{ik}^{(q)}(y_{i}-\hat{m}_{k}^{(q)})(y_{i}-\hat{m}_{k}^{(q)})^{t}$
is the soft covariance matrix of the $k$th group, $\hat{m}_{k}^{(q)}=\frac{1}{n}\sum_{i=1}^{n}t_{ik}^{(q)}y_{i}$
and finally $C=\frac{1}{n}\sum_{k=1}^{K}n_{k}C_{k}$ is the soft within-covariance
matrix of the $K$ groups. \end{prop}

\subsection{Initialization and model selection}

Since the Fisher-EM procedure presented in this work belongs to the
family of EM-based algorithms, the Fisher-EM algorithm can inherit
the most efficient strategies for initialization and model selection
from previous works on the EM algorithm.

\paragraph{Initialization}

Although the EM algorithm is widely used, it is also well-known that
the performance of the algorithm is linked to its initial conditions.
Several strategies have been proposed in the literature for initializing
the EM algorithm. A popular practice~\cite{Biernacki03} executes
the EM algorithm several times from a random initialization and keep
only the set of parameters associated with the highest likelihood.
The use of k-means or of a random partition are also standard approaches
for initializing the algorithm. McLachlan and Peel~\cite{McLachlan00}
have also proposed an initialization through the parameters by generating
the mean and the covariance matrix of each mixture component from
a multivariate normal distribution parametrized by the empirical mean
and empirical covariance matrix of the data. In practice, this latter
initialization procedure works well but, unfortunately, it cannot
be applied directly to the Fisher-EM algorithm since model parameters
live in a space different from the observation space. A simple way
to adapt this strategy could be to first determine a latent space
using PCA and then simulate mixture parameters in this initialization
latent space.

\paragraph{Model selection}

In model-based clustering, it is frequent to consider several models
in order to find the most appropriate model for the considered data.
Since a model is defined by its number of component $K$ and its parametrization,
model selection allows to both select a parametrization and a number
of components. Several criteria for model selection have been proposed
in the literature and the famous ones are penalized likelihood criteria.
Classical tools for model selection include the AIC~\cite{Akaike74},
BIC~\cite{Schwarz78} and ICL~\cite{Biernacki01} criteria. The
Bayesian Information Criterion (BIC) is certainly the most popular
and consists in selecting the model which penalizes the likelihood
by $\frac{\gamma(\mathcal{M})}{2}\log(n)$ where $\gamma(\mathcal{M})$
is the number of parameters in model $\mathcal{M}$ and $n$ is the
number of observations. On the other hand, the AIC criterion penalizes
the log-likelihood by $\gamma(\mathcal{M})$ whereas the ICL criterion
add the penalty $\sum_{i=1}^{n}\sum_{k=1}^{K}t_{ik}\log(t_{ik})$
to the one of the BIC criterion in order to favor well separated models.
The value of $\gamma(\mathcal{M})$ is of course specific to the model
selected by the practitioner (\emph{cf}. Table~\ref{Tab:models}).
In the experiments of the following sections, the BIC criterion is
used because of its popularity but the ICL criterion should also be
well adapted in our context.

\subsection{Computational aspects}

As all iterative procedures, the convergence, the stopping criterion
and the computational cost of the Fisher-EM algorithm deserve to be
discussed.

\paragraph{Convergence}

Although the Fisher-EM algorithm presented in the previous paragraphs
is an EM-like algorithm, it does not satisfy at a first glance to
all conditions required by the convergence theory of the EM algorithm.
Indeed, the update of the orientation matrix~$U$ in the F step is
done by maximizing the Fisher criterion and not by directly maximizing
the expected complete log-likelihood as required in the EM algorithm
theory. From this point of view, the convergence of the Fisher-EM
algorithm cannot therefore be guaranteed. However, as demonstrated
by Campbell~\cite{Campbell84} in the supervised case and by Celeux
and Govaert~\cite{Celeux92} in the unsupervised case, the maximization
of the Fisher criterion is equivalent to the maximization of the complete
likelihood when all mixture components have the same diagonal covariance
matrix ($S_{k}=\sigma^{2}\mathbf{I}_{p}$ for $k=1,...,K$). In our
model, by considering the homoscedastic case with a diagonal covariance
matrix, the conditional expectation of the complete log-likelihood
can be rewritten as $-\frac{n}{2}\left[\mathrm{trace}\left(\left(U^{t}SU\right)^{-1}\left(U^{t}WU\right)\right)\right]+\gamma$
where $\gamma$ is a constant term according to $U$. Hence, with
these assumptions, maximizing this criterion according to $U$ is
equivalent to minimizing the Fisher criterion $\mathrm{trace}\left(\left(U^{t}SU\right)^{-1}\left(U^{t}WU\right)\right)$.
Consequently, for the model $\mathrm{DLM}_{[\alpha\beta]}$ which
assumes the equality and the diagonality of covariance matrices, the
F step of the Fisher-EM algorithm satisfies the convergence conditions
of the EM algorithm theory and the convergence of the Fisher-EM algorithm
can be guaranteed in this case. For the other DLM models, although
the convergence of the Fisher-EM procedure cannot be guaranteed, our
practical experience has shown that the Fisher-EM algorithm rarely
fails to converge with these models if correctly initialized.

\paragraph{Stopping criterion and convergence monitoring}

To decide whether the algorithm has converged or not, we propose to
use the Aitken's criterion~\cite{McLachlan97}. This criterion estimates
the asymptotic maximum of the log-likelihood in order to detect in
advance the algorithm convergence. Indeed, the convergence of the
EM algorithm can be sometimes slow in practice due to its linear convergence
rate and it is often not necessary to wait for the actual convergence
to obtain a good parameter estimate under standard conditions. At
iteration $q$, the Aitken's criterion is defined by $A^{(q)}=\left(\ell^{(q+1)}-\ell^{(q)}\right)/\left(\ell^{(q)}-\ell^{(q-1)}\right)$
where $\ell^{(q)}$ is the log-likelihood value at iteration $q$.
Then, asymptotic estimate of the log-likelihood maximum is given by:\[
\ell_{\infty}^{(q+1)}=\ell^{(q)}+\frac{1}{1-A^{(q)}}\left(\ell^{(q+1)}-\ell^{(q)}\right),\]
and the algorithm can be considered to have converged if $\left|\ell_{\infty}^{(q+1)}-\ell_{\infty}^{(q)}\right|$
is smaller than a small positive number (provided by the user). In
practice, if the criterion is not satisfied after a maximum number
of iterations (provided by the user as well), the algorithm stops.
Afterward, it is possible to check whether the provided estimate is
a local maximum by computing the Hessian matrix (using finite differentiation)
which should be negative definite. In the experiments presented in
the following section, the convergence of the Fisher-EM algorithm
has been checked using such an approach.

\paragraph{Computational cost}

Obviously, since the additional F step is iterative, the computational
complexity of the Fisher-EM procedure is somewhat bigger than the
one of the ordinary EM algorithm. The F step requires $d(d-2)/2$
iterations due to the Gram-Schmidt procedure used for the orthogonalization
of $U$. However, since $d$ is at most equal to $K-1$ and is supposed
to be small compared to $p$, the complexity of the F step is not
a quadratic function of the data dimension which could be large. Furthermore,
it is important to notice that the complexity of this step does not
depend on the number of observations $n$. Although the proposed algorithm
is more time consuming than the usual EM algorithm, it is altogether
actually usable on recent PCs even for large scale problems. Indeed,
we have observed on simulations that Fisher-EM appears to be $1.5$
times slower on average than EM (with a diagonal model). As an example,
$24$ seconds are on average necessary for Fisher-EM to cluster a
dataset of $1\,000$ observations in a $100$-dimensional space whereas
EM requires $16$ seconds.

\subsection{Practical aspects\label{par_n<<p}}

The DLM models, for which the Fisher-EM algorithm has been proposed
as an estimation procedure, presents several practical and numerical
interests among which the ability to visualize the clustered data,
to interpret the discriminative axes and to deal with the so-called
$n\ll p$ problem.

\paragraph{Choice of $d$ and visualization in the discriminative subspace}

The proposed DLM models are all parametrized by the intrisinc dimension
$d$ of the discriminative latent subspace which is theoretically
at most equal to $K-1$. Even though the actual value of $d$ is strictly
smaller than $K-1$ for the dataset at hand, we recommand in practice
to set $d=K-1$ when numerically possible in order to avoid stability
problems with the Fisher-EM algorithm. Furthermore, it is always better
to extract more discriminative axes than to miss relevant dimensions
and $K-1$ is often in practice a small value compared to $p$. Besides,
a natural use of the discriminative axes may certainly be the visualization
of the clustered data. Indeed, it is nowadays clear that the visualization
help human operators to understand the results of an analysis. With
the Fisher-EM algorithm, it is easy to project and visualize the cluster
data into the estimated discriminative latent subspace if $K\leq4$.
When $K>4$, the actual value of $d$ can be estimated by looking
at the eigenvalue scree of $S_{W}^{-1}S_{B}$ and two cases have therefore
to be considered. On the one hand, if the estimated value of $d$
is at most equal to $3$, the practitioner can therefore visualize
his data by projecting them on the $d$ first discriminative axes
and no discriminative information loss is to be deplored in this case.
On the other hand, if the estimated value of $d$ is strictly larger
than $3$, the visualization becomes obviously more difficult but
the practitioner may simply use the $3$ first discriminative axes
which are the most discriminative ones among the $K-1$ provided axes.
Let us finally notice that the visualization quality is of course
related to the clustering quality. Indeed, the visualization provided
by the Fisher-EM algorithm may be disappointing if the clustering
results are poor, due to a bad initialization for instance. A good
solution to avoid such a situation may be to initialize the Fisher-EM
algorithm with the {}``mini-EM'' strategy or with the results of
a classical EM algorithm.

\paragraph*{Interpretation of the discriminative axes}

Beyond the natural interest of visualization, it may also be useful
from a practical point of view to interpret the estimated discriminative
axes, \emph{i.e.} $u_{1},...,u_{d}$ with the notations of the previous
sections. The main interest for the practitioner would be to figure
out which original dimensions are the most discriminative. This can
be done by looking at the matrix $U$ which contains $u_{1},...,u_{d}$
as column vectors. In the classical framework of factor analysis,
this matrix is known as the loading matrix (the discriminative axes
$u_{1},...,u_{d}$ are the loadings). Thus, it is possible to find
the most discriminative original variables by selecting the highest
values in the loadings. A simple way to highlight the relevant variables
is to threshold the loadings (setting to zero the values less than
a given threshold). Let us finally remark that finding the most discriminative
original variables is of particular interest in application fields,
such as biology or economics, where the observed variables have an
actual meaning.

\paragraph{Dealing with the $n\ll p$ problem}

Another important and frequent problem when clustering high-dimensional
data is known as high dimension and low sample size (HDSS) problem
or the $n\ll p$ problem (we refer to \cite[Chap. 18]{Hastie01} for
an overview). The $n\ll p$ problem refers to situations where the
number of features $p$ is larger than the number of available observations~$n$.
This problem occurs frequently in modern scientific applications such
as genomics or mass spectrometry. In such cases, the estimation of
model parameters for generative clustering methods is either difficult
or impossible. This task is indeed very difficult when $n\ll p$ since
generative methods require, in particular, to invert covariance matrices
which are ill-conditioned in the best case or singular in the worst
one. In contrast with other generative methods, the Fisher-EM procedure
can overcome the $n\ll p$ problem. Indeed, the E and M steps of Fisher-EM
do not require the determination of the last $p-d$ columns of $W$
(see equations~(\ref{Eq_Estep}) and~(\ref{Eq_Mstep_debut-1})--(\ref{DUPLICATION : Eq_Mstep_fin}))
and, consequently, it is possible to modify the F step to deal with
situations where $n\ll p$. To do so, let $\bar{Y}$ denote the centered
data matrix and $T$ denote, as before, the soft partition matrix.
We define in addition the weighted soft partition matrix $\tilde{T}$
where the $j$th column $\tilde{T}_{j}$ of $\tilde{T}$ is the $j$th
column $T_{j}$ of $T$ divided by $n_{j}=\sum_{i=1}^{n}t_{ij}$.
With these notations, the between covariance matrix $B$ can be written
in its matrix form $B=\bar{Y}^{t}\tilde{T}^{t}\tilde{T}\bar{Y}$ and
the F step aims to maximize, under orthogonality constraints, the
function $f(U)=\mathrm{trace}\left((U^{t}\bar{Y}^{t}\bar{Y}U)^{-1}U^{t}\bar{Y}^{t}\tilde{T}^{t}\tilde{T}\bar{Y}U\right).$
It follows from the classical result of kernel theory, the Representer
theorem~\cite{Kimeldorf71}, that this maximization can be done in
a different space and that $U$ can be expressed as $U=\bar{Y}H$
where $H\in\mathbb{R}^{n\times p}$. Therefore, the F step reduces
to maximize, under orthogonality constraints, the following function:
\begin{equation}
f(H)=\mathrm{trace}\left((H^{t}GGH)^{-1}H^{t}G\tilde{T}^{t}\tilde{T}GH\right),\label{KernTrick2}\end{equation}
 where $G=\bar{Y}\bar{Y}^{t}$ is the $n\times n$ Gram matrix. The
solution $U^{*}$ of the original problem can be obtained afterward
from the solution $H^{*}$ of~(\ref{KernTrick2}) by multiplying
it by $\bar{Y}$. Thus, the F step reduces to the eigendecomposition
under orthogonality constraints of a $n\times n$ matrix instead of
a $p\times p$ matrix. This procedure is useful for the Fisher-EM
procedure only because it allows to determine $d\leq n$ axes which
are enough for Fisher-EM but not for other generative methods which
require the computation of the $p$ axes.

\section{Experimental results\label{sec:Experimental-results}}

This section presents experiments on simulated and real datasets in
order to highlight the main features of the clustering method introduced
in the previous sections.

\subsection{An introductory example: the Fisher's irises\label{sub:An-introductory-example:}}

Since we chose to name the clustering algorithm proposed in this work
after Sir R.~A.~Fisher, the least we can do is to first apply the
Fisher-EM algorithm to the iris dataset that Fisher used in~\cite{Fisher36}
as an illustration for his discriminant analysis. This dataset, in
fact collected by E.~Anderson~\cite{Anderson35} in the Gaspé peninsula
(Canada), is made of three groups corresponding to different species
of iris (\textit{setosa}, \textit{versicolor} and \textit{virginica})
among which the groups \textit{versicolor} and\textit{ virginica}
are difficult to discriminate (they are at least not linearly separable).
The dataset consists of 50 samples from each of three species and
four features were measured from each sample. The four measurements
are the length and the width of the sepal and the petal. This dataset
is used here as an introductory example because of the link with Fisher's
work but also of its popularity in the clustering community.

\begin{figure}[p]
\centering{}\begin{tabular}{cc}
\includegraphics[scale=0.4]{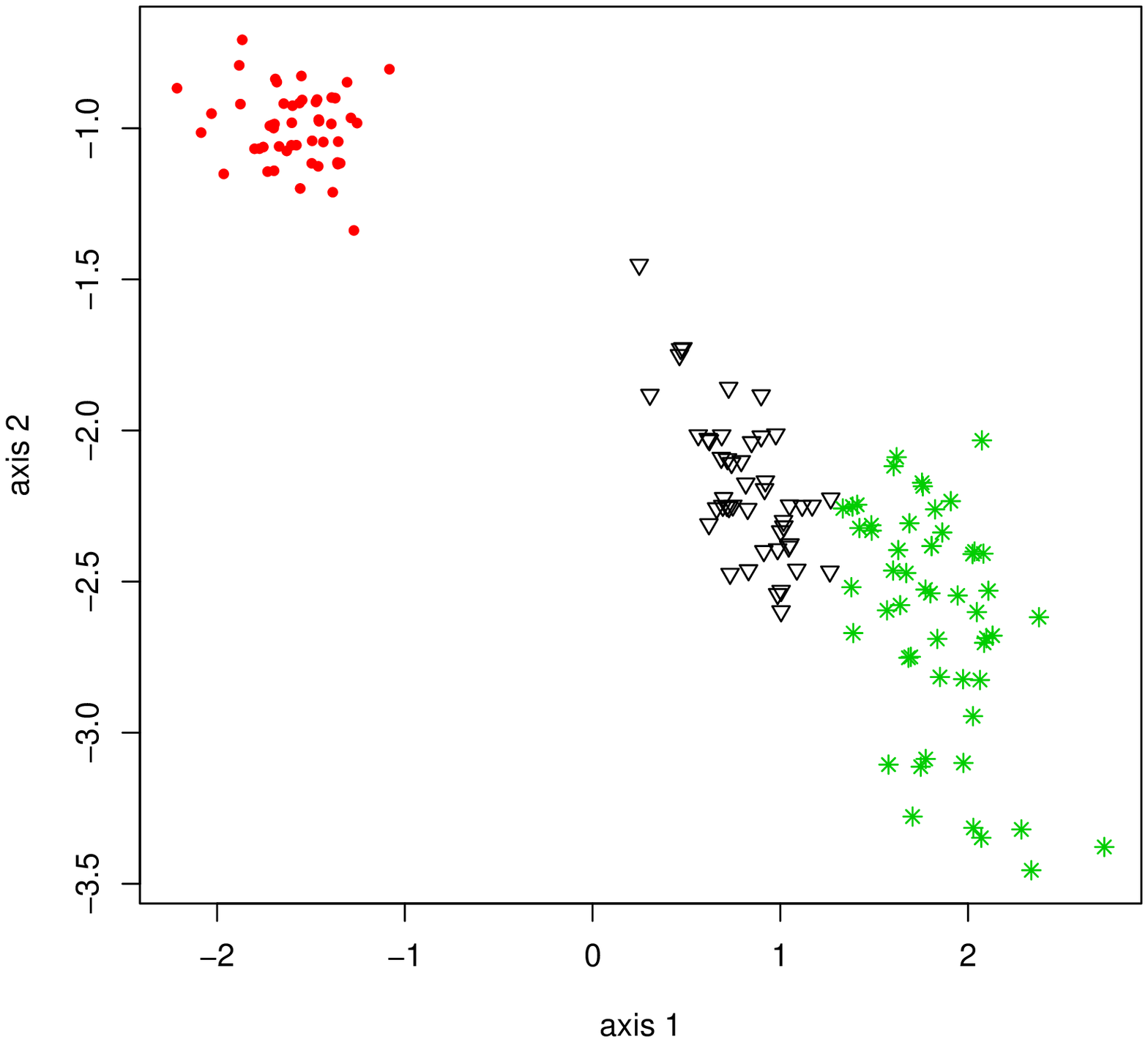} &  \includegraphics[scale=0.4]{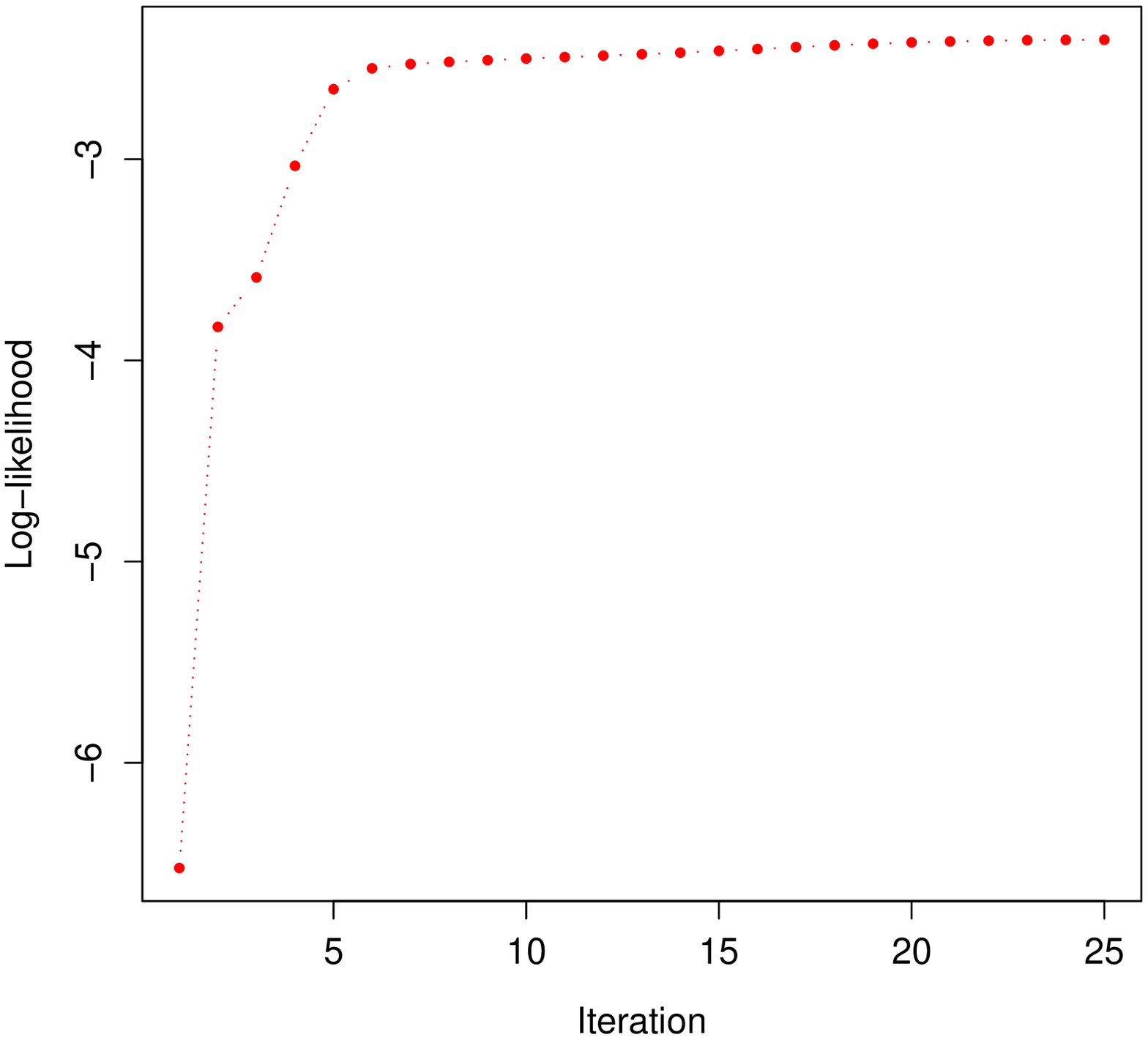}\tabularnewline
\end{tabular} \caption{\label{iris:data} Projection of clustered Iris data into the latent
discriminative subspace with Fisher-EM (left) and evolution of the
associated log-likelihood (right).}

\end{figure}
\begin{table}[p]
\begin{centering}
\begin{tabular}{cc}
\begin{tabular}{|lccc|lccc|}
\hline 
\multicolumn{4}{|c|}{{\small OLDA}} & \multicolumn{4}{c|}{Fisher-EM}\tabularnewline
\hline
\hline 
 & \multicolumn{3}{c|}{\textit{\small cluster}} &  & \multicolumn{3}{c|}{\textit{\small cluster}}\tabularnewline
\textit{\small class} & {\small 1} & {\small 2} & {\small 3} & \textit{\small class} & {\small 1} & {\small 2} & {\small 3}\tabularnewline
\hline
{\small Setosa} & {\small 50} & {\small 0} & {\small 0} & {\small Setosa} & {\small 50} & {\small 0} & {\small 0}\tabularnewline
{\small Versicolor} & {\small 0} & {\small 48} & {\small 2} & {\small Versicolor} & {\small 0} & {\small 47} & {\small 3}\tabularnewline
{\small Virginica} & {\small 0} & {\small 1} & {\small 49} & {\small Virginica} & {\small 0} & {\small 0} & {\small 50}\tabularnewline
\hline 
\multicolumn{4}{|l|}{\textit{\scriptsize Misclassification rate = 0.02}} & \multicolumn{4}{c|}{\textit{\scriptsize Misclassification rate = 0.02}}\tabularnewline
\hline
\end{tabular} & \tabularnewline
\end{tabular}
\par\end{centering}

\caption{\label{iris:comp} Confusion tables for the iris data with OLDA method
(supervised) and Fisher-EM (unsupervised).}

\end{table}

\begin{table}
\begin{centering}
\begin{tabular}{|l|cc|cc|}
\cline{2-5} 
\multicolumn{1}{l|}{} & \multicolumn{2}{c|}{OLDA} & \multicolumn{2}{c|}{Fisher-EM}\tabularnewline
\hline 
 & \multicolumn{2}{c|}{\textit{axis}} & \multicolumn{2}{c|}{\textit{axis}}\tabularnewline
\textit{variable} & {\small 1} & {\small 2} & {\small 1} & {\small 2}\tabularnewline
\hline
sepal length & 0.209 & 0.044 & -0.203 & -0.108\tabularnewline
sepal width & 0.386  & 0.665 & -0.422 & 0.088\tabularnewline
petal length & -0.554  & -0.356 & 0.602 & 0.736\tabularnewline
petal width & -0.707 & 0.655 & 0.646 & -0.662\tabularnewline
\hline
\end{tabular}
\par\end{centering}

\caption{\label{tab:Fisher-axes-estimated}Fisher axes estimated in the supervised
case (OLDA) and in the unsupervised case (Fisher-EM).}

\end{table}
In this first experiment, Fisher-EM has been applied to the iris data
(of course, the labels have been used only for performance evaluation)
and the Fisher-EM results will be compared to the ones obtained in
the supervised case with the orthogonal linear analysis method (OLDA)~\cite{Ye05}.
The left panel of Figure~\ref{iris:data} stands for the projection
of the irises in the estimated discriminative space with Fisher-EM
and the right panel shows the evolution of the log-likelihood on $25$
iterations until convergence. First of all, it can be observed that
the estimated latent space discriminates almost perfectly the three
different groups. For this experiment, the clustering accuracy has
reached $98\%$ with the DLM$_{[\alpha_{k}\beta]}$ model of Fisher-EM.
Secondly, the right panel shows the monotonicity of the evolution
of the log-likelihood and the convergence of the algorithm to a stationary
state. Table~\ref{iris:comp} presents the confusion matrices for
the partitions obtained with supervised and unsupervised classification
methods. OLDA has been used for the supervised case (reclassification
of the learning data) whereas Fisher-EM has provided the clustering
results. One can observe that the obtained partitions induced by both
methods is almost the same. This confirms that Fisher-EM has correctly
modeled both the discriminative subspace and the groups within the
subspace. It is also interesting to look at the loadings provided
by both methods. Table~\ref{tab:Fisher-axes-estimated} stands for
the linear coefficients of the discriminative axes estimated, on the
one hand, in the supervised case (OLDA) and, on the other hand, in
the unsupervised case (Fisher-EM). The first axes of each approach
appear to be very similar and the scalar product of these axes is
$-0.996$. This highlights the performance of the Fisher-EM algorithm
in estimating the discriminative subspace of the data. Furthermore,
according to these results, the $3$ groups of irises can be mainly
discriminated by the petal size meaning that only one axis would be
sufficient to discriminate the $3$ iris species. Besides, this interpretation
turns out to be in accordance with the recent work of Trendafilov
and Joliffe~\cite{Trendafilov07} on variable selection in discriminant
analysis via the LASSO.

\subsection{Simulation study: influence of the dimension}

\begin{figure}[t]
\centering{}\begin{tabular}{cc}
\includegraphics[width=0.49\columnwidth]{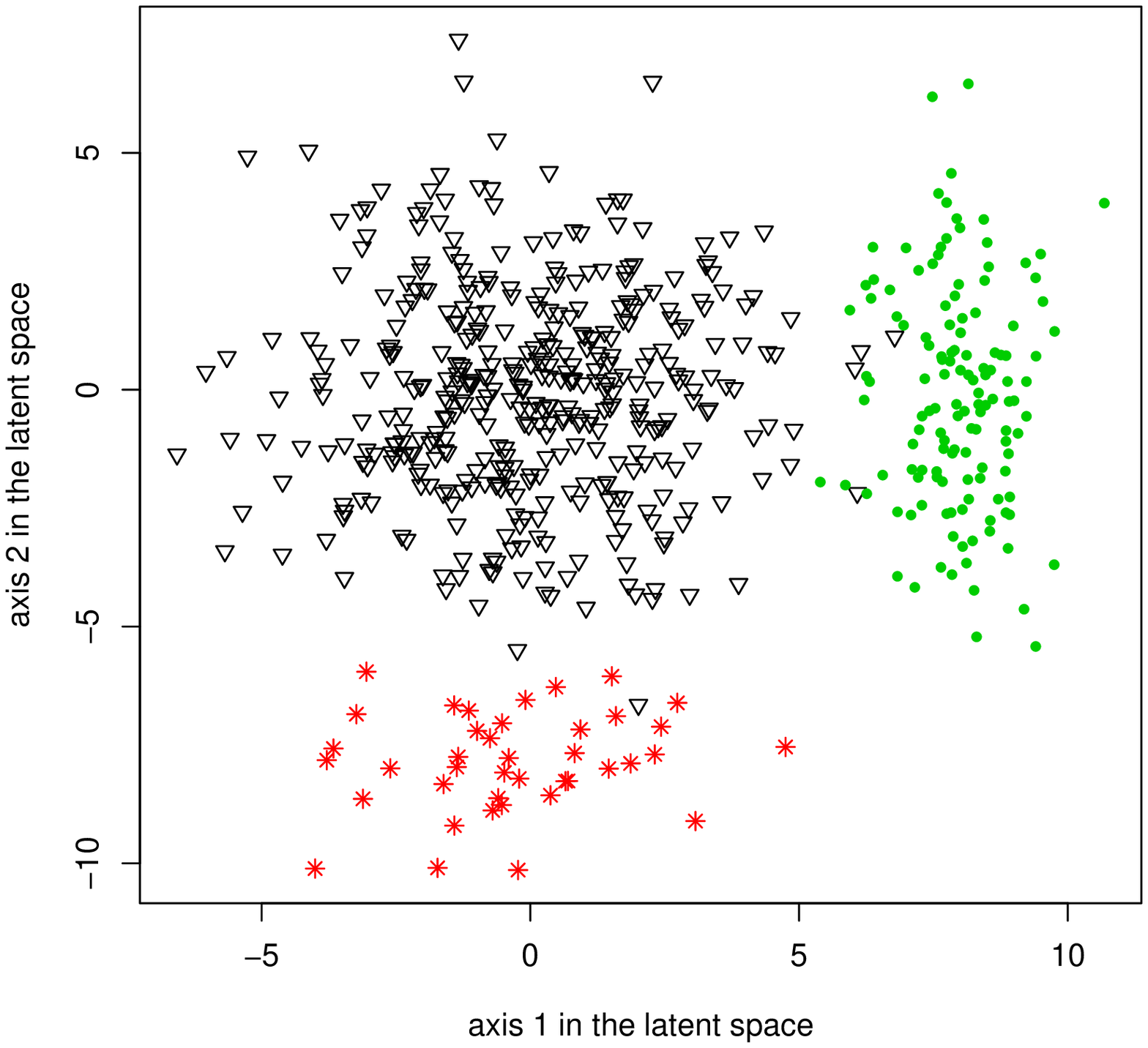}  & \includegraphics[width=0.49\columnwidth]{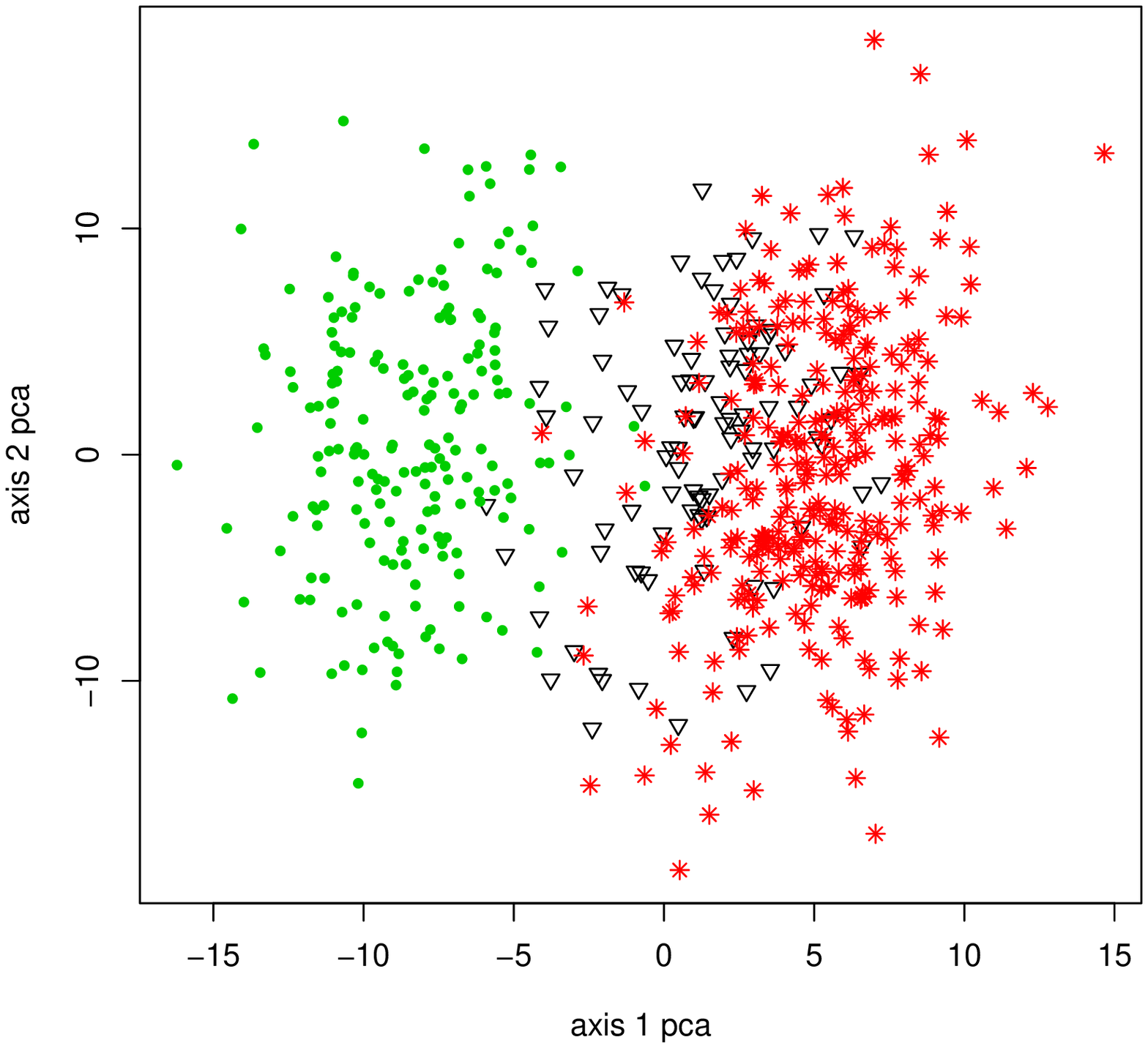} \tabularnewline
\end{tabular}\caption{\label{Fig:simul_dimens} Visualization of the simulated data: data
in their latent space (left) and data projected on the first principal
components (right).}

\end{figure}

\begin{figure}[t]
\centering{}\includegraphics[width=0.65\columnwidth,height=0.4\textheight]{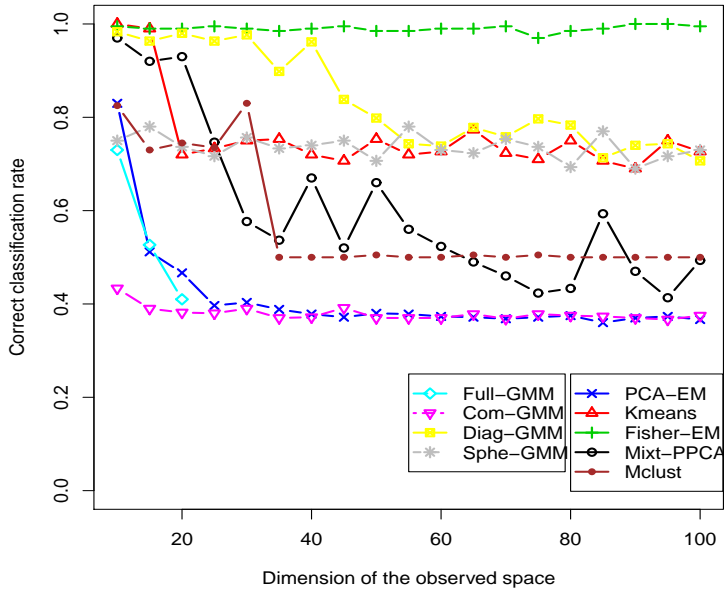}
\caption{\label{Fig:bxp} Influence of the dimension of the observed space
on the correct classification rate for Full-GMM, PCA-EM, Com-GMM,
Mixt-PPCA, k-means, Diag-GMM, Sphe-GMM and Fisher-EM algorithms.}

\end{figure}

\begin{figure}[p]
\begin{centering}
\begin{tabular}{ccc}
\includegraphics[width=0.32\columnwidth]{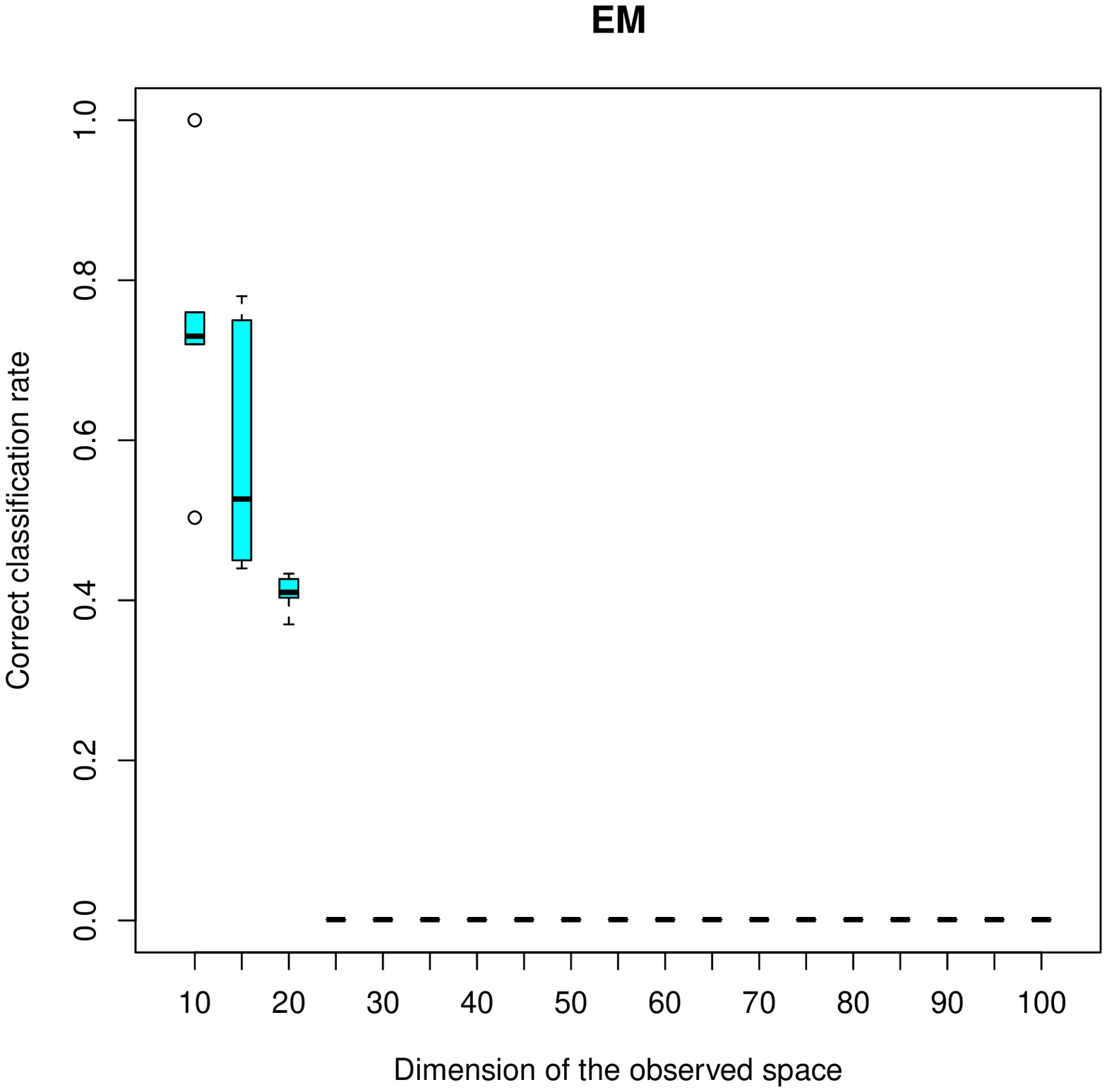}  & \includegraphics[width=0.32\columnwidth]{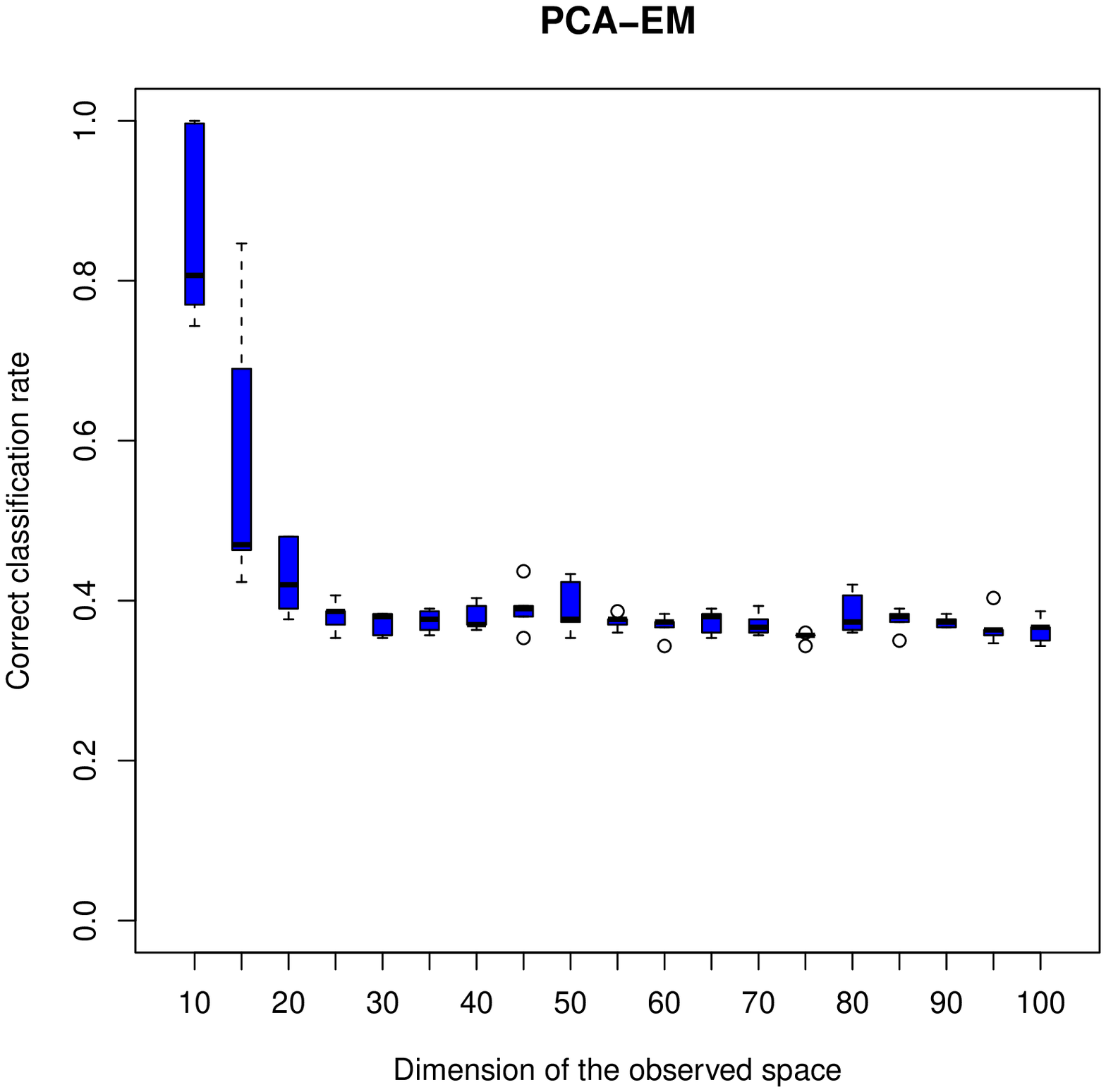}  & \includegraphics[width=0.32\columnwidth]{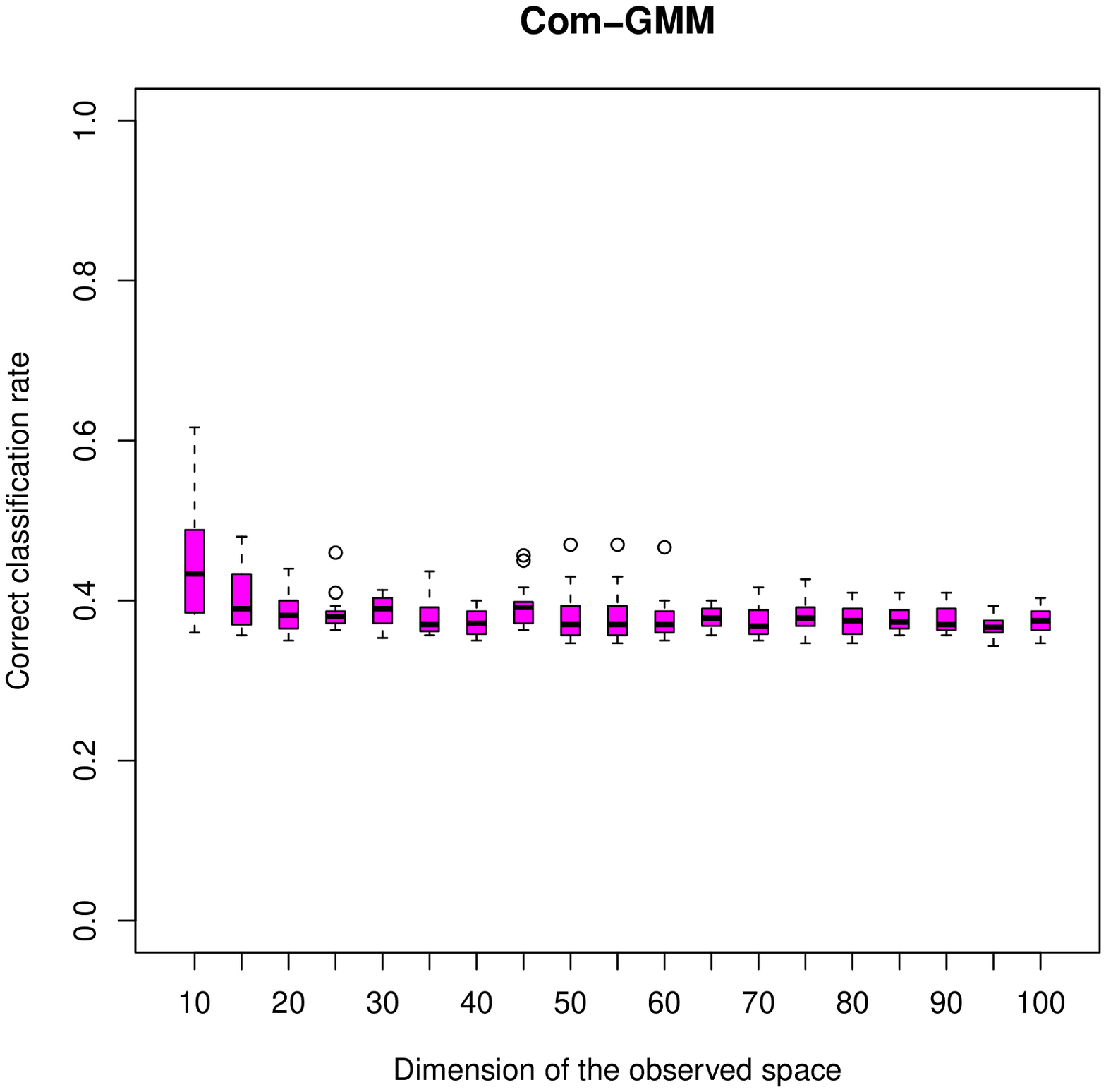}\tabularnewline
\end{tabular}\\
\begin{tabular}{ccc}
\includegraphics[width=0.32\columnwidth]{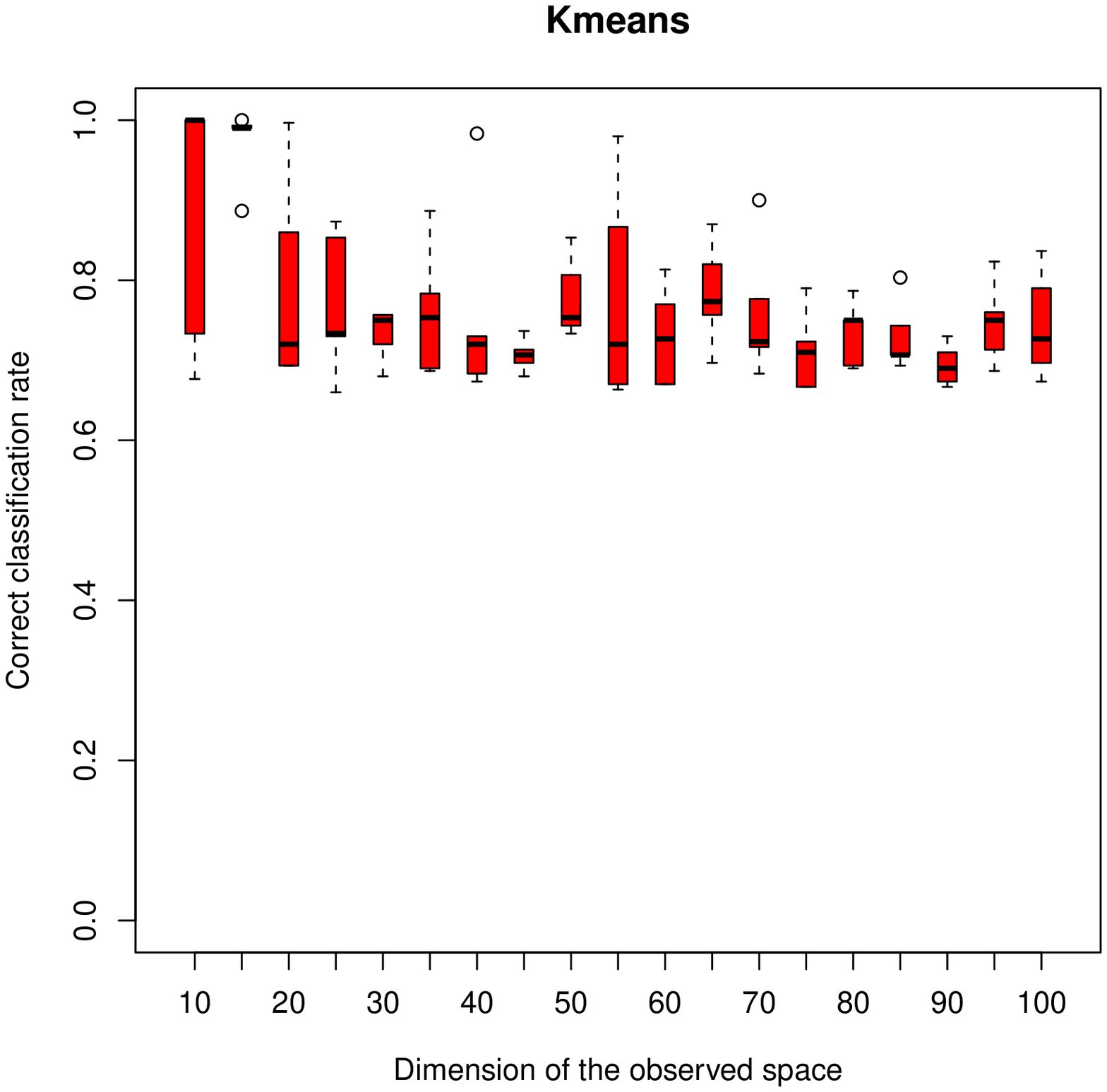} & \includegraphics[width=0.32\columnwidth]{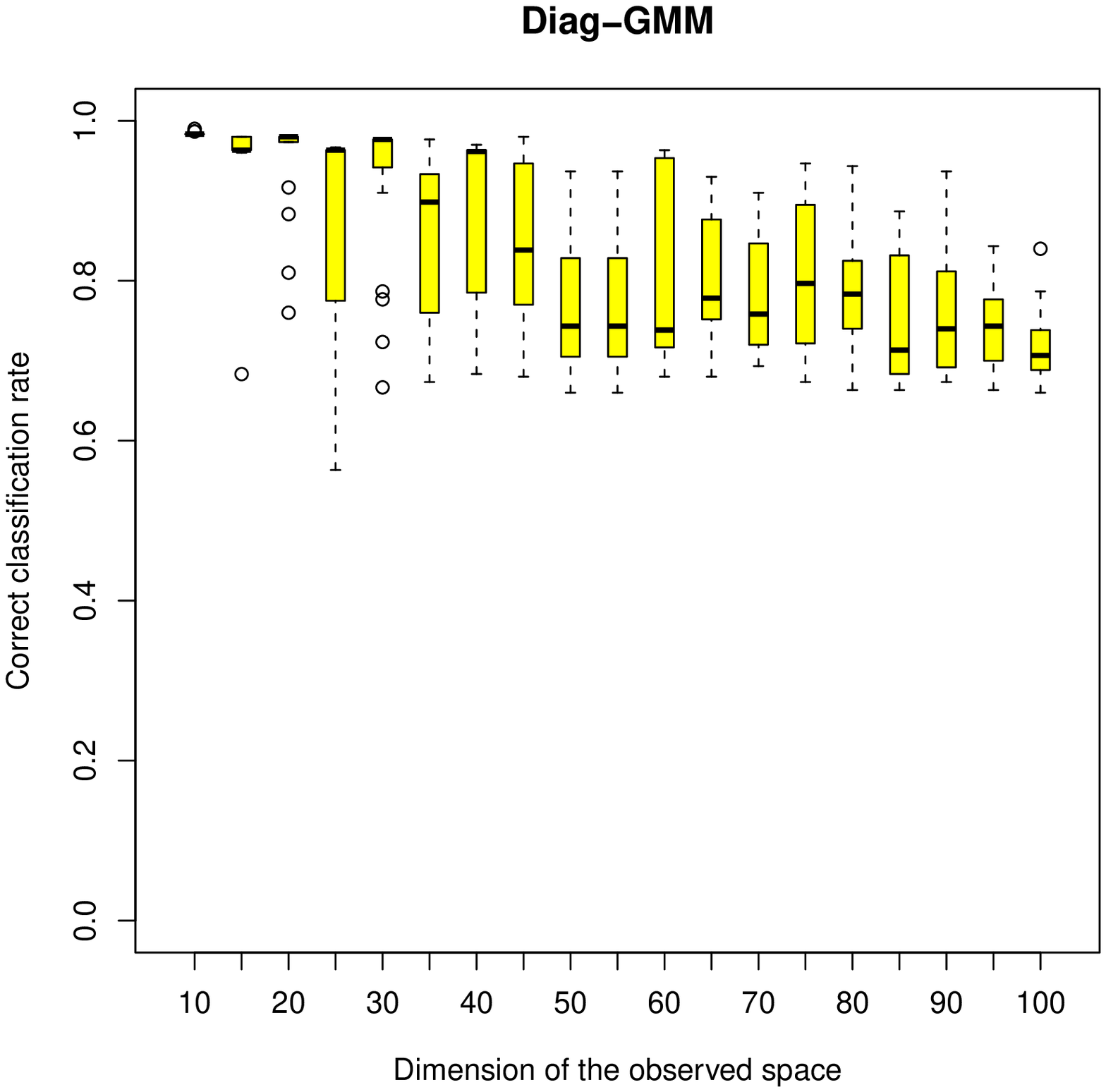} & \includegraphics[width=0.32\columnwidth]{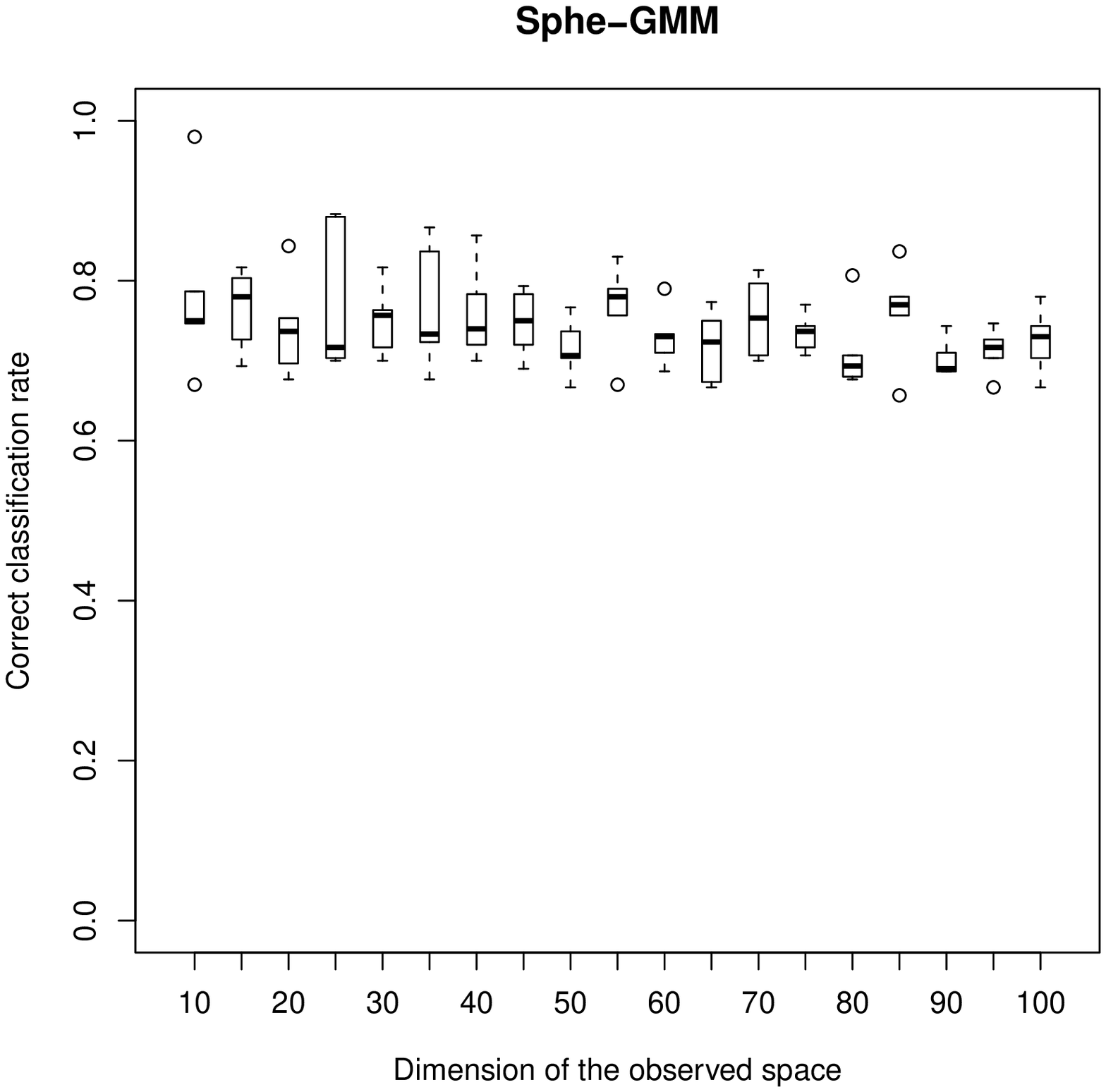}\tabularnewline
\includegraphics[width=0.32\columnwidth]{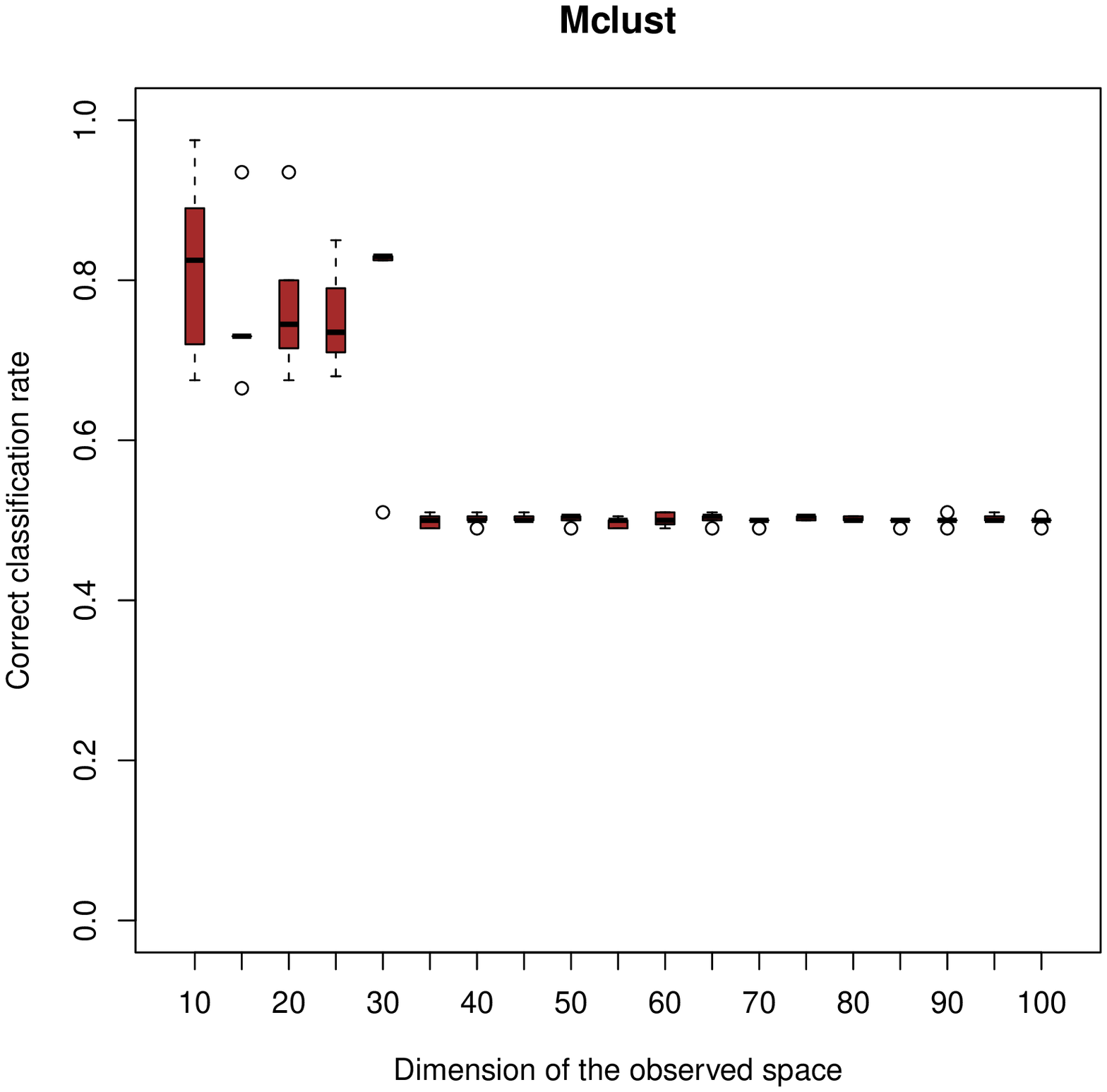} & \includegraphics[width=0.32\columnwidth]{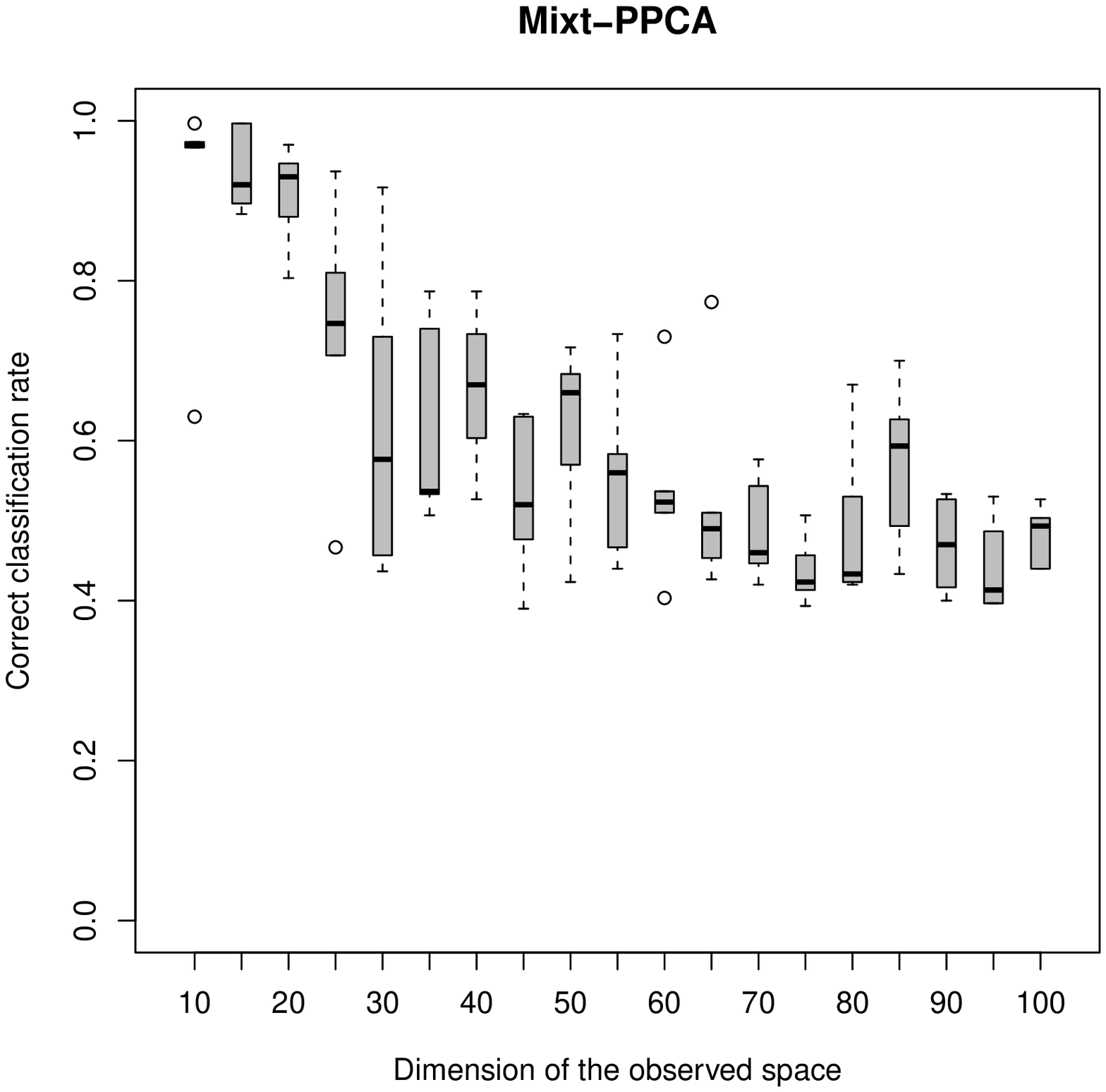} & \includegraphics[width=0.32\columnwidth]{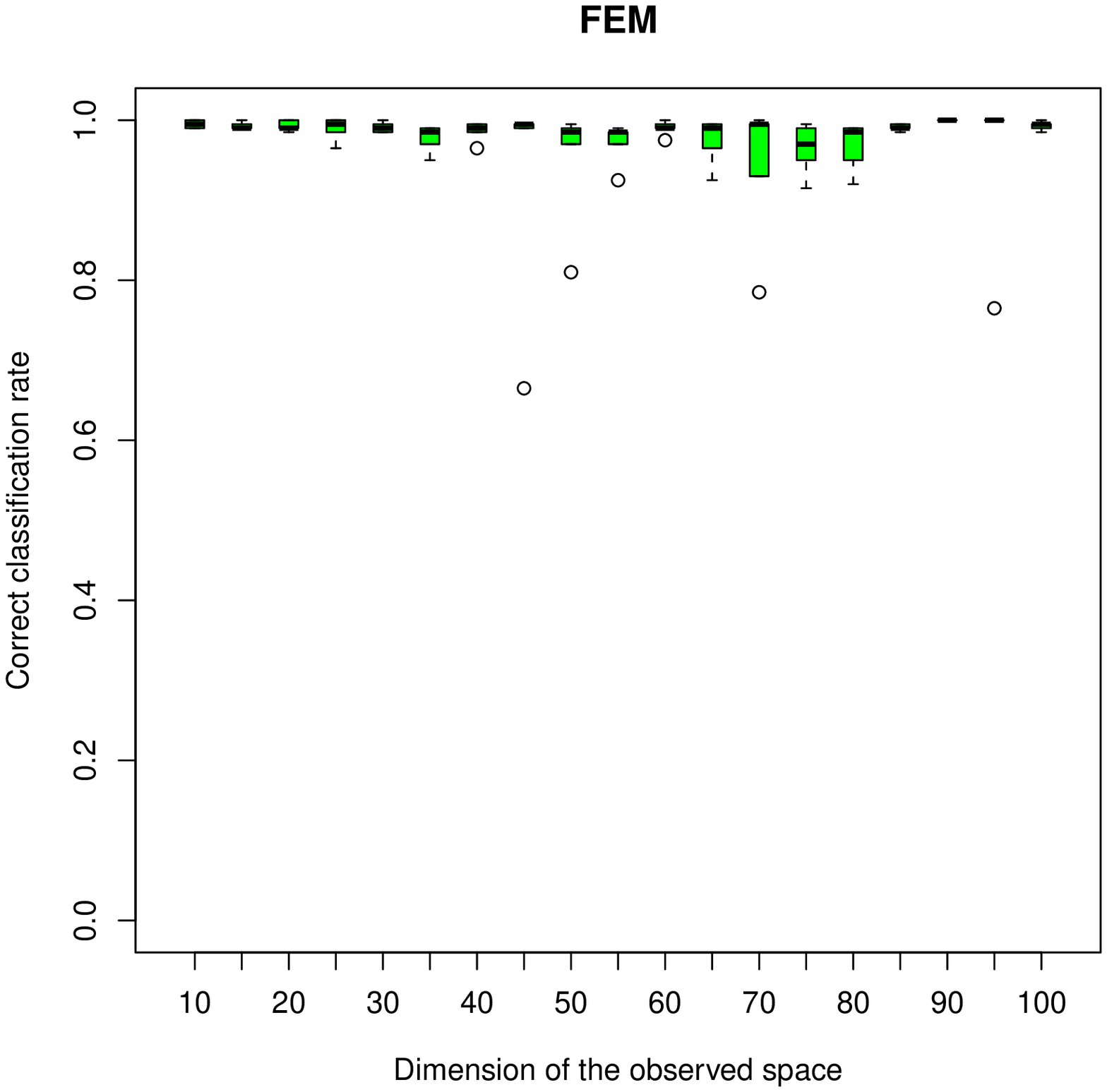}\tabularnewline
\end{tabular}
\par\end{centering}

\centering{}\caption{\label{Fig:bxp2} Boxplots of Full-GMM, PCA-EM, Com-GMM, Mixt-PPCA,
k-means, Diag-GMM, Sphe-GMM and Fisher-EM algorithms.}

\end{figure}

This second experiment aims to compare with traditional methods the
stability and the efficiency of the Fisher-EM algorithm in partitioning
high-dimensional data. Fisher-EM is compared here with the standard
EM algorithm (Full-GMM) and its parsimonious models (Diag-GMM, Sphe-GMM
and Com-GMM), the EM algorithm applied in the first components of
PCA explaining $90\%$ of the total variance (PCA-EM), the k-means
algorithm and the mixture of probabilistic principal component analyzers
(Mixt-PPCA). For this simulation, $600$ observations have been simulated
following the $\mathrm{DLM}_{[\alpha_{kj}\beta_{k}]}$ model proposed
in Section~\ref{sec:Model-DLM}. The simulated dataset is made of
$3$ unbalanced groups and each group is modeled by a Gaussian density
in a $2$-dimensional space completed by orthogonal dimensions of
Gaussian noise. The transformation matrix $W$ has been randomly simulated
such as $W^{t}W=WW^{t}=I_{p}$ and, for this experience, the dimension
of the observed space varies from $5$ to $100$. The left panel of
Figure~\ref{Fig:simul_dimens} shows the simulated data in their
$2$-dimensional latent space whereas the right panel presents the
projection of $50$-dimensional observed data on the two first axes
of PCA in the observed space. As one can observe, the representation
of the data on the two first principal components is actually not
well suited for clustering these data while it exists a representation
which discriminates perfectly the three groups. Moreover, to make
the results of each method comparable, the same randomized initialization
has been used for the $8$ algorithms. The experimental process has
been repeated $20$ times for each dimension of the observed space
in order to see both the average performances and their variances.
Figure~\ref{Fig:bxp} presents the evolution of the clustering accuracy
of each method (EM, PCA-EM, k-means, Mixt-PPCA, Fisher-EM, Diag-GMM,
Sphe-GMM and Com-GMM) according to the data dimensionality and Figure~\ref{Fig:bxp2}
presents their respective boxplots. First of all, it can be observed
that the Full-GMM, PCA-EM and Com-GMM have their performances which
decrease quickly when the dimension increases. In fact, the Full-GMM
model does not work upon the $15$th dimension and still remains unstable
in a low dimensional space as well as the Com-GMM model. Similarly,
the performances of PCA-EM fall down as the $10$th dimension. This
can be explained by the fact that the latent subspace provided by
PCA does not allow to well discriminate the groups, as already suggested
by Figure~\ref{Fig:simul_dimens}. However, the PCA-EM approach can
be used whatever the dimension is whereas Full-GMM cannot be used
as the $20$th dimension because of numerical problems linked to singularity
of the covariance matrices. Moreover, their boxplots show a large
variation on the clustering accuracy. Secondly, Sphe-GMM, Diag-GMM
and k-means present the same trend with high performances in low-dimensional
spaces which decrease until they reach a clustering accuracy of $0.75$.
Diag-GMM seems however to resist a little bit more than k-means to
the dimension increasing. Mixt-PPCA and Mclust both follow the same
tendency as the previous methods but from the $30$th dimension their
performances fall down until the clustering accuracy reaches $0.5$.
The poor performances of Mixt-PPCA can be explained by the fact that
Mixt-PPCA models each group in a different subspace whereas the model
used for simulating the observations assumes a common discriminative
subspace. Finally, Fisher-EM appears to be more effective than the
other methods and, more importantly, it remains very stable while
the data dimensionality increases. Furthermore, the boxplot associated
with the Fisher-EM results suggests that it is a steady algorithm
which succeeds in finding out the discriminative latent subspace of
the data even with random initializations.

\subsection{Simulation study: model selection}

This last experiment on simulations aims to study the performance
of BIC for both model and component number selection. For this experiment,
$4$ Gaussian components of $75$ observations each have been simulated
according to the DLM$_{[\alpha_{k}\beta]}$ model in a $3$-dimensional
space completed by 47 orthogonal dimensions of Gaussian noise (the
dimension of the observation space is therefore $p=50$). The transformation
matrix $W$ has been again randomly simulated such as $W^{t}W=WW^{t}=I_{p}$.
Table \ref{Tab:BIC} presents the BIC values for the family of DLM
models and, in a comparative purpose, the BIC values for 7 other methods
already used in the last experiments: EM with the Full-GMM, Diag-GMM,
Sphe-GMM and Com-GMM models, Mixt-PPCA, Mclust~\cite{Fraley99} (with
model {[}EEE{]} which is the most appropriate model for these data)
and PCA-EM. Moreover, BIC is computed for different partition numbers
varying between 2 and 6 clusters. First of all, one can observe that
the BIC values linked to the models which are different from the DLM
model are very low compared to the DLM models. This suggests that
the models which best fit the data are the DLM models. Secondly, 8
of the 12 DLM models select the right number of components ($K=4$).
In particular, the DLM models which assume a common variance between
each cluster outside the latent subspace (models DLM$_{[.\beta]}$)
all select the 4 clusters. The other methods under-estimate the number
of clusters. BIC has the largest value for the DLM$_{[\alpha_{k}\beta]}$
model with $4$ components which is actually the model used for simulating
the data. Finally, the right-hand side of Figure \ref{fig:Projection_3D}
presents the projection of the data on the discriminative subspace
of $3$ dimensions estimated by Fisher-EM with the DLM$_{[\alpha_{k}\beta]}$
model whereas the left-hand side figure represents the projection
of the data on the $3$ first principal components of PCA. As one
can observe, in the PCA case, the axes separate only 2 groups, which
is in accordance with the model selection pointed out by BIC for this
method. Conversely, in the Fisher-EM case, the 3 discriminative axes
separate well the $4$ groups and such a representation could clearly
help the practitioner in understanding the clustering results. 

\begin{table}[p]
\begin{centering}
\begin{tabular}{|l|ccccc|}
\cline{2-6} 
\multicolumn{1}{l|}{} & \multicolumn{5}{c|}{{\small number of components}}\tabularnewline
\hline
{\small methods} & 2 & 3 & 4 & 5 & 6\tabularnewline
\hline
\hline 
{\small DLM$_{[\Sigma_{k}\beta_{k}]}$ } & -114.6172 & -114.5996 & -115.4875 & -115.6439 & -116.7350\tabularnewline
{\small DLM$_{[\Sigma_{k}\beta]}$ } & -116.9006 & -117.4791 & -115.0215  & -116.0837 & -116.8912\tabularnewline
{\small DLM$_{[\Sigma\beta_{k}]}$ } & -116.9007 & -116.9568 & -118.5480  & -119.3458 & -120.0418\tabularnewline
{\small DLM$_{[\Sigma\beta]}$ } & -120.9006 & -120.2496 & -119.8787  & -120.6301 & -120.6166\tabularnewline
DLM$_{[\alpha_{kj}\beta_{k}]}$  & -116.5750 & -114.9578 & -114.7986  & -115.6658 & -116.5750\tabularnewline
DLM$_{[\alpha_{kj}\beta]}$  & -121.8565 & -117.4968 & -115.1525 & -115.8571  & -117.7598\tabularnewline
DLM$_{[\alpha_{k}\beta_{k}]}$  & -115.2290 & -115.0808 & -114.7934  & -115.6603 & -116.5027\tabularnewline
DLM$_{[\alpha_{k}\beta]}$  & -121.8565 & -117.6217 & \textbf{-114.1471 } & -115.7909 & -116.6739\tabularnewline
DLM$_{[\alpha_{j}\beta_{k}]}$  & -116.7295 & -118.4031 & -119.2610  & -120.7783 & -122.0415\tabularnewline
DLM$_{[\alpha_{j}\beta]}$  & -123.3448 & -120.9052 & -120.4578  & -121.1248 & -121.9098\tabularnewline
DLM$_{[\alpha\beta_{k}]}$  & -118.7295 & -118.3865 & -119.7309  & -121.5124 & -123.1506\tabularnewline
DLM$_{[\alpha\beta]}$  & -123.3443 & -120.8989 & -120.4347  & -121.7451 & -123.2730\tabularnewline
\hline
Full-GMM  & -177.6835 & -252.8908 & -440.6805  & -3005.531 & -4367.653\tabularnewline
Com-GMM  & -150.0518 & -193.0624 & -231.4546  & -270.2741 & -309.7809\tabularnewline
Mixt-PPCA  & -151.1561 & -176.3615 & -201.5709  & -226.7789 & -251.9931\tabularnewline
Diag-GMM & -189.8663 & -262.7929 & -419.360  & -407.2755  & -466.6955 \tabularnewline
Sphe-GMM & -190.9812 & -258.3534 & -302.8030  & -382.7666 & -433.3845\tabularnewline
\hline
PCA-EM & -127.0857 & -173.7174 & -247.3894  & -364.9811 & -594.4000\tabularnewline
Mclust$_{[EII]}$ & -229.3360 & -229.3024 & -230.0155  & -230.8431 & -231.5140\tabularnewline
\hline
\end{tabular}
\par\end{centering}

\caption{\label{Tab:BIC} BIC values for model selection.}

\end{table}
\begin{figure}[p]
\begin{centering}
\begin{tabular}{cc}
\textbf{\tiny Projection on the 3 first principal components} & \textbf{\tiny Projection on the discriminative axes estimated by Fisher-EM}\tabularnewline
\includegraphics[bb=0bp 20bp 504bp 440bp,clip,width=0.49\columnwidth]{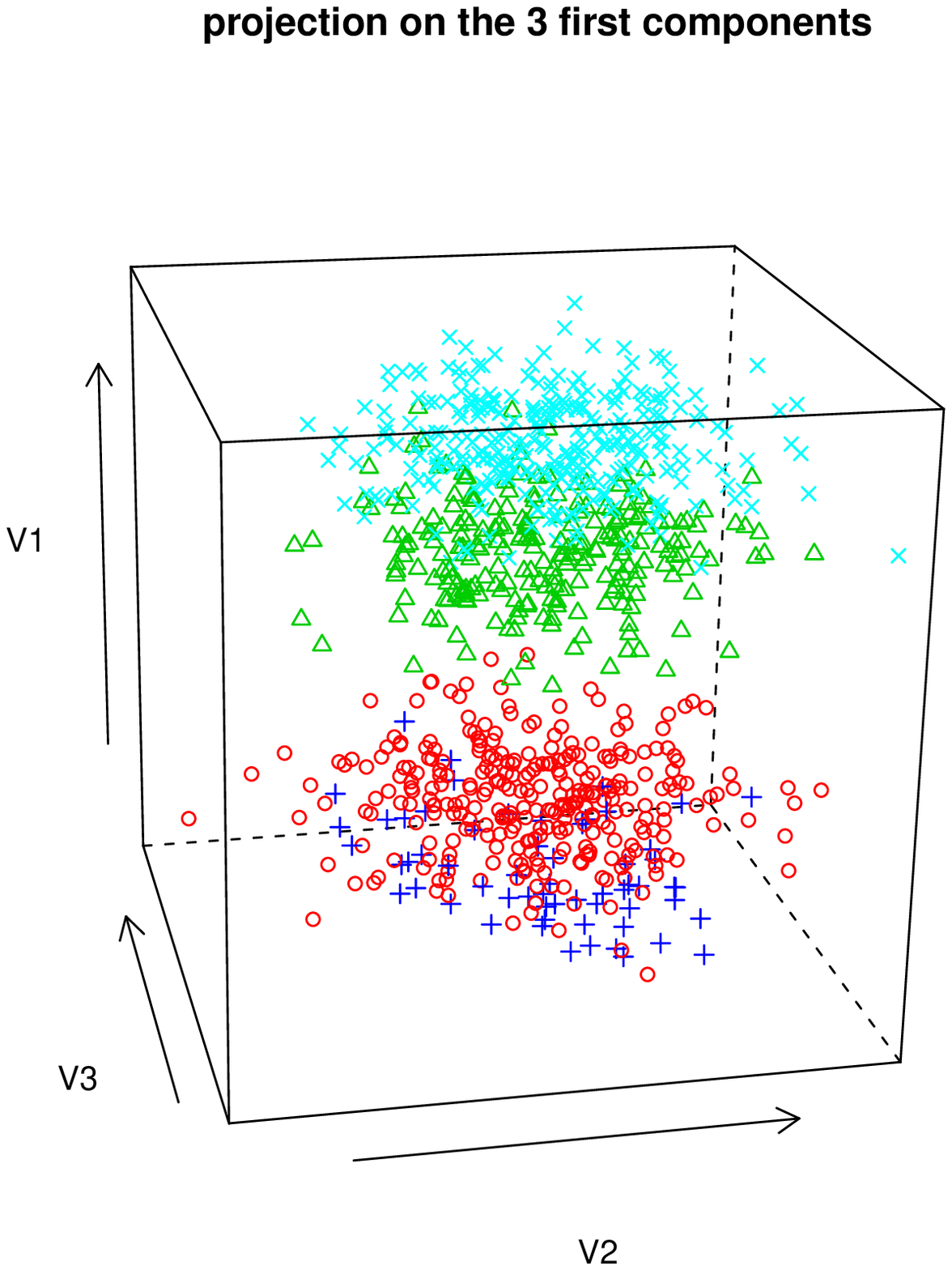} & \includegraphics[bb=0bp 20bp 504bp 440bp,clip,width=0.49\columnwidth]{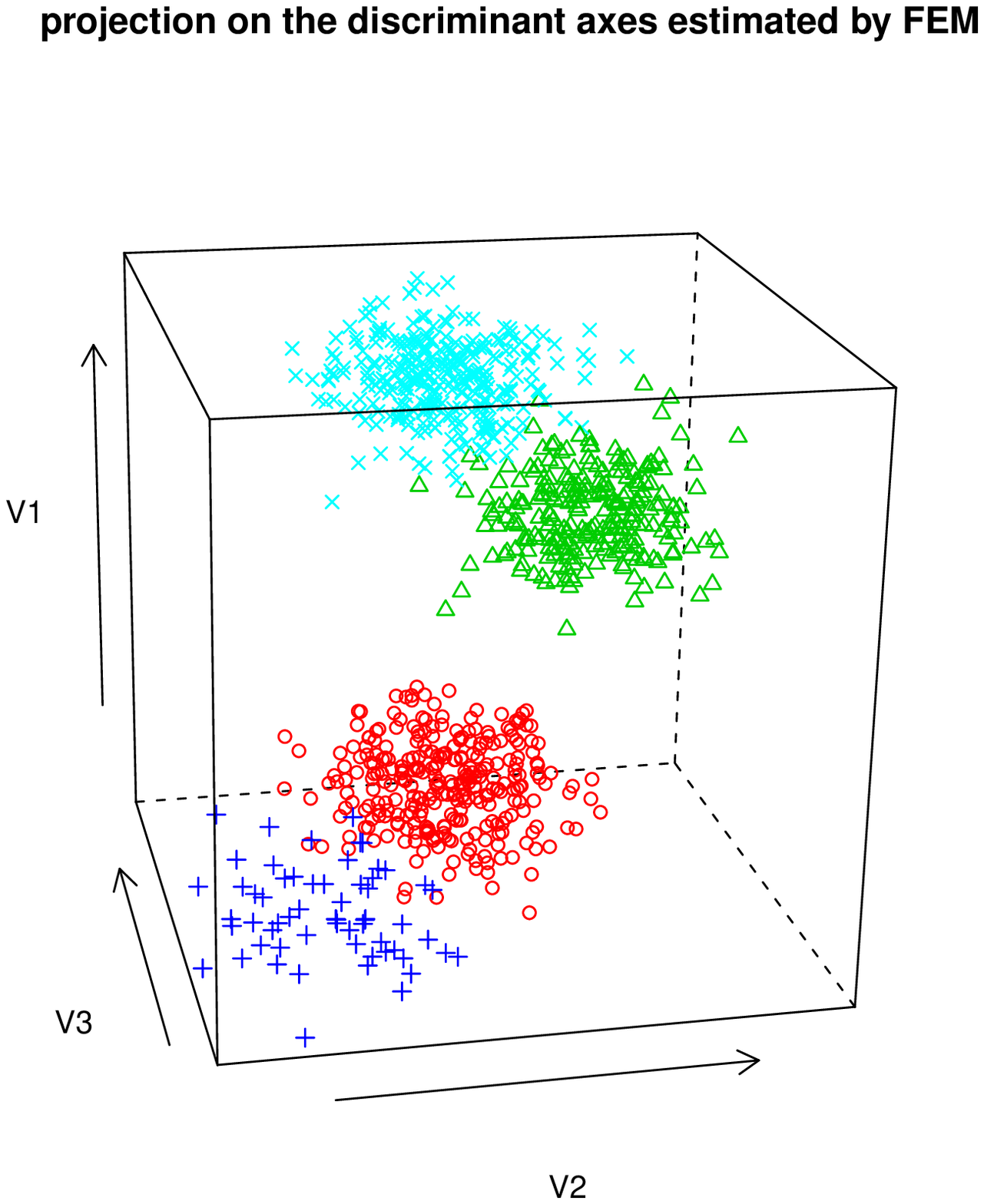}\tabularnewline
\end{tabular}
\par\end{centering}

\caption{\label{fig:Projection_3D}Projection of the data in the 3 first principal
components of PCA (left) and in the discriminant subspace estimated
by Fisher-EM with the DLM$_{[\alpha_{k}\beta]}.$}

\end{figure}

\subsection{Real data set benchmark}

This last experimental paragraph will focus on comparing on real-world
datasets the efficiency of Fisher-EM with several linear and nonlinear
existing methods, including the most recent ones. On the one hand,
Fisher-EM will be compared to the 8 already used clustering methods:
EM with the Full-GMM, Diag-GMM, Sphe-GMM and Com-GMM models, Mixt-PPCA,
Mclust (with its most adapted model for these data), PCA-EM and k-means.
On the other hand, the new Fisher-EM challengers will be k-means computed
on the two first components of PCA (PCA--k-means), an heteroscedastic
factor mixture analyzer (HMFA) method~\cite{Montanari2010} and three
discriminative versions of k-means: LDA--k-means~\cite{Ding07},
Dis--k-means and DisCluster (see~\cite{Ye07} for more details).
The comparison has been made on $7$ different benchmark datasets
coming mostly from the UCI machine learning repository: 
\begin{itemize}
\item The \textbf{chironomus} data contain $148$ larvae which are split
up into $3$ species and described by $17$ morphometric attributes.
This dataset is described in detailed in~\cite{Montanari2010}. 
\item The \textbf{wine} dataset is composed by $178$ observations which
are split up into $3$ classes and characterized by $13$ variables. 
\item The \textbf{iris} dataset which is made of $3$ different groups and
described by $4$ variables. This dataset has been described in detail
in Section~\ref{sub:An-introductory-example:}. 
\item The \textbf{zoo} dataset includes $7$ families of $101$ animals
characterized by $16$ variables. 
\item The \textbf{glass} data are composed by $214$ observations belonging
to $6$ different groups and described by $7$ variables. 
\item The $4435$ \textbf{satellite images} are split up into $6$ classes
and are described by $36$ variables. 
\item Finally, the last dataset is the\textbf{ USPS data} where only the
classes which are difficult to discriminate are considered. Consequently,
this dataset consists of 1756 records (rows) and 256 attributes divided
in $3$ classes (numbers 3, 5 and 8).
\end{itemize}
Table~\ref{Tab:UCI} presents the average clustering accuracies and
the associated standard deviations obtained for the $12$ DLM models
and for the methods already used in the previous experiments. The
results for the $19$ first methods of the table have been obtained
by averaging $20$ trials with random initializations except for Mclust
which has its own deterministic initialization and this explains the
lack of standard deviation for Mclust. Similarly, Table~\ref{tab:UCI-2}
provides the clustering accuracies found in the literature for the
recent methods on the same datasets. It is important to notice that
the results of Table~\ref{tab:UCI-2} have been obtained in slightly
different benchmarking situations. Missing values in Table~\ref{Tab:UCI}
are due to non-convergence of the algorithms whereas missing values
in Table~\ref{tab:UCI-2} are due to the unavailability of the information
for the concerned method. First of all, one can remark that Fisher-EM
outperforms the other methods for most of the UCI datasets such as
wine, iris, zoo, glass, satimage and usps358 datasets. Finally, it
is interesting from a practical point of view to notice that some
DLM models work well in most situations. In particular, the DLM$_{[.\beta]}$
models, in which the variance outside the discriminant subspace is
common to all groups, provide very satisfying results for all the
datasets considered here.

\begin{table}[p]
\noindent \begin{raggedright}
\begin{tabular}{|l|c|c|c|c|c|c|c|}
\hline 
{\small Method} & {\small iris } & {\small wine } & {\small chironomus} & {\small zoo } & {\small glass } & {\small satimage } & {\small usps358}\tabularnewline
\hline
\hline 
{\small DLM$_{[\Sigma_{k}\beta_{k}]}$ } & {\small 94.8$\pm2.3$} & {\small 96.1$\pm0.0$} & {\small 91.7$\pm5.2$} & {\small -} & {\small 39.5}\textbf{\small $\pm1.8$} & {\small 64.6$\pm2.2$} & {\small 77.9$\pm7.1$}\tabularnewline
{\small DLM$_{[\Sigma_{k}\beta]}$ } & {\small 96.7$\pm0.0$} & {\small 95.5$\pm0.0$} & {\small 97.2$\pm0.1$} & {\small -} & {\small 39.9}\textbf{\small $\pm1.4$} & {\small 65.7$\pm0.8$} & {\small 70.0$\pm8.5$}\tabularnewline
{\small DLM$_{[\Sigma\beta_{k}]}$ } & {\small 81.9$\pm2.4$} & {\small 94.1$\pm1.3$} & {\small 91.8$\pm2.4$} & {\small 73.3$\pm5.5$} & {\small 40.6}\textbf{\small $\pm0.9$} & {\small 62.7$\pm1.9$} & {\small 74.1$\pm9.4$}\tabularnewline
{\small DLM$_{[\Sigma\beta]}$ } & {\small 77.8$\pm3.7$ } & {\small 93.6$\pm1.6$} & {\small 89.1$\pm6.3$} & {\small 78.4$\pm6.4$} & {\small 38.5}\textbf{\small $\pm1.9$} & \textbf{\small 68.0}{\small $\pm1.7$} & {\small 66.4$\pm8.7$}\tabularnewline
{\small DLM$_{[\alpha_{kj}\beta_{k}]}$ } & {\small 89.3$\pm0.0$} & {\small 95.5$\pm0.0$} & {\small 86.1$\pm6.3$} & {\small 73.7$\pm3.5$} & \textbf{\small 42.0$\pm2.2$} & {\small 65.5$\pm2.0$} & {\small 74.8$\pm9.1$}\tabularnewline
{\small DLM$_{[\alpha_{kj}\beta]}$ } & {\small 91.1$\pm1.4$} & {\small 94.2$\pm0.2$} & {\small 96.3$\pm7.0$} & {\small 70.4$\pm5.3$} & {\small 40.1$\pm3.3$} & {\small 65.0$\pm2.9$} & {\small 68.7$\pm11.1$}\tabularnewline
{\small DLM$_{[\alpha_{k}\beta_{k}]}$ } & {\small 96.1$\pm2.2$} & {\small 95.5$\pm0.0$} & {\small 87.5$\pm3.9$} & {\small 73.7$\pm3.6$} & {\small 39.2$\pm3.7$} & {\small 64.4$\pm2.1$} & {\small 76.2$\pm7.6$}\tabularnewline
{\small DLM$_{[\alpha_{k}\beta]}$ } & \textbf{\small 98}{\small .0}\textbf{\small $\pm0.0$} & {\small 94.3$\pm0.0$} & {\small 96.2$\pm6.8$} & {\small 72.8$\pm3.1$} & {\small 40.1$\pm2.0$} & {\small 58.9$\pm5.3$} & {\small 74.1$\pm10.6$}\tabularnewline
{\small DLM$_{[\alpha_{j}\beta_{k}]}$ } & {\small 79.3$\pm3.6$} & {\small 93.8$\pm2.8$} & {\small 83.7$\pm3.9$} & {\small 72.5$\pm7.0$} & {\small 39.4$\pm0.9$} & {\small 62.4$\pm1.8$} & {\small 77.8$\pm8.2$}\tabularnewline
{\small DLM$_{[\alpha_{j}\beta]}$ } & {\small 72.7$\pm6.5$} & {\small 92.6$\pm3.2$} & {\small 89.7$\pm6.3$} & \textbf{\small 80}{\small .1}\textbf{\small $\pm4.2$} & {\small 39.5$\pm1.5$} & \textbf{\small 68.0$\pm1.5$} & {\small 74.2$\pm11.2$}\tabularnewline
{\small DLM$_{[\alpha\beta_{k}]}$ } & {\small 80.3$\pm4.3$} & {\small 96.3$\pm1.9$} & {\small 83.6$\pm8.5$} & {\small 70.2$\pm7.0$} & {\small 39.1$\pm2.4$} & {\small 62.4$\pm2.5$} & \textbf{\small 81.2}{\small $\pm6.5$}\tabularnewline
{\small DLM$_{[\alpha\beta]}$ } & {\small 79.8$\pm4.0$} & \textbf{\small 97.1$\pm0.0$} & {\small 89.8$\pm6.6$} & {\small 78.0$\pm4.8$} & {\small 38.4$\pm1.3$} & {\small 67.9}\textbf{\small $\pm1.3$} & {\small 72.8$\pm9.8$}\tabularnewline
\hline 
{\small Full-GMM } & {\small 79.0$\pm5.7$} & {\small 60.9$\pm7.7$} & {\small 44.8$\pm4.1$} & {\small -} & {\small 38.3$\pm2.1$} & {\small 35.9$\pm3.1$} & {\small -}\tabularnewline
{\small Com-GMM} & {\small 57.6$\pm18.3$} & {\small 61.0$\pm14.9$} & {\small 51.9$\pm10.9$} & {\small 59.9$\pm10.3$} & {\small 38.3$\pm3.1$} & {\small 26.1$\pm1.5$} & {\small 38.2$\pm1.1$}\tabularnewline
{\small Mixt-PPCA } & {\small 89.1$\pm4.2$} & {\small 63.1$\pm7.9$} & {\small 56.3$\pm4.5$} & {\small 50.9$\pm6.5$} & {\small 37.0$\pm2.3$} & {\small 40.6$\pm4.7$} & {\small 53.1$\pm9.6$}\tabularnewline
{\small Diag-GMM} & {\small 93.5$\pm1.3$} & {\small 94.6$\pm2.8$} & {\small 92.1$\pm4.2$} & {\small 70.9$\pm12.3$} & {\small 39.1$\pm2.4$} & {\small 60.8$\pm5.2$} & {\small 45.9$\pm9.1$}\tabularnewline
{\small Sphe-GMM} & {\small 89.4$\pm0.4$} & {\small 96.6$\pm0.0$} & {\small 85.9$\pm9.9$} & {\small 69.4$\pm5.4$} & {\small 37.0$\pm2.1$} & {\small 60.2$\pm7.5$} & {\small 78.7$\pm11.2$}\tabularnewline
{\small PCA-EM } & {\small 66.9$\pm9.9$} & {\small 64.4$\pm5.7$} & {\small 66.1$\pm4.0$ } & {\small 61.9$\pm6.2$} & {\small 39.0$\pm1.7$} & {\small 56.2$\pm4.2$} & {\small 67.6$\pm11.2$}\tabularnewline
{\small k-means } & {\small 88.7$\pm4.0$} & {\small 95.9$\pm4.0$} & {\small 92.9$\pm6.0$} & {\small 68.0$\pm7.4$} & {\small 41.3$\pm2.8$} & {\small 66.6$\pm4.1$} & {\small 74.9$\pm13.9$}\tabularnewline
\hline 
{\small Mclust} & {\small 96.7} & \textbf{\small 97.1 } & \textbf{\small 97.9} & {\small 65.3} & {\small 41.6} & {\small 58.7} & {\small 55.5}\tabularnewline
\textit{\footnotesize Model name} & \textit{\small (VEV)} & \textit{\small (VVI)} & \textit{\small (EEE)} & \textit{\small (EII)} & \textit{\small (VEV)} & \textit{\small (VVV)} & \textit{\small (EEE)}\tabularnewline
\hline
\end{tabular}
\par\end{raggedright}

\caption{\label{Tab:UCI}Clustering accuracies and their standard deviations
(in percentage){\small{} }on the UCI datasets averaged on 20 trials.
No standard deviation is reported for Mclust since its initialization
procedure is deterministic and always provides the same initial partition.}

\end{table}

\begin{table}[p]
\centering{}\begin{tabular}{|l|c|c|c|c|c|c|c|}
\hline 
{\small Method} & {\small wine } & {\small iris } & {\small chironomus } & {\small zoo } & {\small glass } & {\small satimage } & {\small usps358}\tabularnewline
\hline
\hline 
{\small PCA--k-means~\cite{Ding07} } & {\small 70.2 } & {\small 88.7 } & {\small - } & {\small 79.2 } & {\small 47.2 } & {\small - } & {\small -}\tabularnewline
\hline 
{\small LDA--k-means~\cite{Ding07}} & {\small 82.6 } & {\small 98.0}\textbf{\small{} } & {\small - } & {\small 84.2}\textbf{\small{} } & {\small 51.0}\textbf{\small{} } & {\small - } & {\small -}\tabularnewline
\hline 
{\small Dis--k-means~\cite{Ye07}} & {\small - } & {\small - } & {\small - } & {\small - } & {\small - } & {\small 65.1} & {\small -}\tabularnewline
\hline 
{\small DisCluster~\cite{Ye07}} & {\small - } & {\small - } & {\small - } & {\small - } & {\small - } & {\small 64.2} & {\small -}\tabularnewline
\hline 
{\small HMFA~\cite{Montanari2010}} & {\small - } & {\small - } & {\small 98.7}\textbf{\small{} } & {\small - } & {\small - } & {\small -} & {\small -}\tabularnewline
\hline
\end{tabular}\caption{\label{tab:UCI-2}Clustering accuracies (in percentage) on UCI datasets
found in the literature (these results have been obtained with slightly
different experimental setups). }

\end{table}

\section{Application to mass spectrometry\label{sec:Application-to-mass-spectrometry}}

In this last experimental section, the Fisher-EM procedure is applied
to the problem of cancer detection using MALDI mass spectrometry.
MALDI mass spectrometry is a non-invasive biochemical technique which
is useful in searching for disease biomarkers, assessing tumor progression
or evaluating the efficiency of drug treatment, to name just a few
applications. In particular, a promising field of application is the
early detection of the colorectal cancer, which is one of the principal
causes of cancer-related mortality, and MALDI imaging could in few
years avoid in some cases the colonoscopy method which is invasive
and quite expensive.

\subsection{Data and experimental setup}

\begin{figure}
\begin{centering}
\includegraphics[bb=0bp 300bp 576bp 525bp,clip,width=0.8\columnwidth]{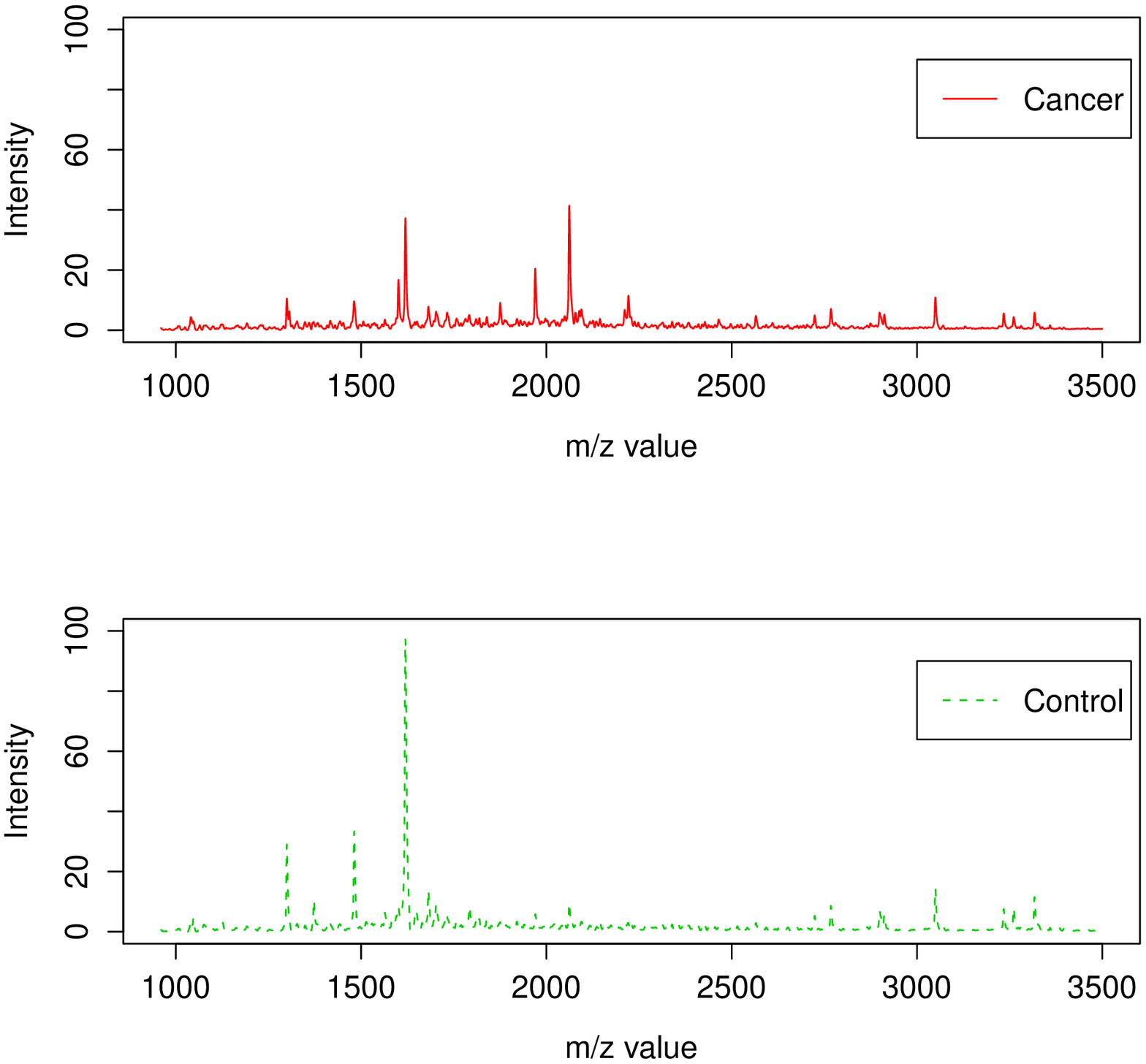}
\par\end{centering}

\begin{centering}
\includegraphics[bb=0bp 15bp 576bp 235bp,clip,width=0.8\columnwidth]{images/Maldi_colmeans}
\par\end{centering}

\caption{\label{fig:Maldi-1}Estimated mean spectra of the cancer class (up)
and of the control class (bottom) on the m/z interval 900--3500 Da.}

\end{figure}

The MALDI2009 dataset has been provided by Theodore Alexandrov from
the Center for Industrial Mathematics (University of Bremen, Germany)
and is made of 112 spectra of length 16~331. Among the 112 spectra,
64 are spectra from patients with the colorectal cancer (referred
to as cancer hereafter) and 48 are spectra from healthy persons (referred
to as control). Each of the 112 spectra is a high-dimensional vector
of 16~331 dimensions which covers the mass-to-charge (m/z) ratios
from 960 to 11~163 Da. For further reading, the dataset is presented
in detail and analyzed in a supervised classification framework in~\cite{Alexandrov09}. 

Following the experimental protocol of~\cite{Alexandrov09}, Fisher-EM
was applied on the 6~168 dimensions corresponding to m/z ratios between
960 and 3~500 Da since there is no discriminative information on
the reminder. Figure~\ref{fig:Maldi-1} shows the mean spectra of
the cancer and control classes estimated by Fisher-EM on the m/z interval
900--3500 Da. To be able to compare the clustering results of Fisher-EM,
PCA-EM and mixture of PPCA (Mixt-PPCA) have been applied to this subset
as well. It has been asked to all methods to cluster the dataset into
2 groups. It is important to remark that this clustering problem is
a $n\ll p$ problem and, among the model-based methods, only these
three methods are able to deal with it (see Section~\ref{par_n<<p}).

\subsection{Experimental results}

\begin{table}
\begin{centering}
{\scriptsize }\begin{tabular}{|lcc|}
\hline 
\multicolumn{3}{|c|}{{\scriptsize PCA-EM}}\tabularnewline
\hline
\hline 
 & \multicolumn{2}{c|}{\emph{\scriptsize Cluster}}\tabularnewline
\emph{\scriptsize Class} & {\scriptsize Cancer} & {\scriptsize Control}\tabularnewline
\hline 
{\scriptsize Cancer} & 48 & 16\tabularnewline
{\scriptsize Control} & 1 & 47\tabularnewline
\hline
\multicolumn{3}{|l|}{\emph{\scriptsize Misclassification rate = 0.15}}\tabularnewline
\hline
\end{tabular}{\scriptsize ~~~}\begin{tabular}{|lcc|}
\hline 
\multicolumn{3}{|c|}{{\scriptsize Mixt-PPCA}}\tabularnewline
\hline
\hline 
 & \multicolumn{2}{c|}{\emph{\scriptsize Cluster}}\tabularnewline
\emph{\scriptsize Class} & {\scriptsize Cancer} & {\scriptsize Control}\tabularnewline
\hline 
{\scriptsize Cancer} & 62 & 2\tabularnewline
{\scriptsize Control} & 10 & 38\tabularnewline
\hline
\multicolumn{3}{|l|}{\emph{\scriptsize Misclassification rate = 0.11}}\tabularnewline
\hline
\end{tabular}{\scriptsize ~~~}\begin{tabular}{|lcc|}
\hline 
\multicolumn{3}{|c|}{{\scriptsize Fisher-EM}}\tabularnewline
\hline
\hline 
 & \multicolumn{2}{c|}{\emph{\scriptsize Cluster}}\tabularnewline
\emph{\scriptsize Class} & {\scriptsize Cancer} & {\scriptsize Control}\tabularnewline
\hline 
{\scriptsize Cancer} & 57 & 7\tabularnewline
{\scriptsize Control} & 3 & 45\tabularnewline
\hline
\multicolumn{3}{|l|}{\emph{\scriptsize Misclassification rate = 0.09}}\tabularnewline
\hline
\end{tabular}
\par\end{centering}{\scriptsize \par}

\caption{\label{tab:Maldi}Confusion tables for PCA-EM (left), mixture of PPCA
(center) and Fisher-EM (right).}

\end{table}

Table~\ref{tab:Maldi} presents the confusion tables computed from
the clustering results of PCA-EM, mixture of PPCA and Fisher-EM. On
the one hand, PCA-EM has selected $d=4$ principal axes with the 90\%
variance rule before to cluster the data in this subspace and mixture
of PPCA has selected $d=2$ principal axes for each group. On the
other hand, Fisher-EM has estimated the discriminative latent subspace
with $d=K-1=1$ axis to cluster this high-dimensional dataset. It
first appears that PCA-EM and mixture of PPCA provide satisfying clustering
results on such a complex dataset. However, it is disappointing to
see that the PCA-EM make a significant number of false negatives (cancers
classified as non-cancers) since the classification risk is not symmetric
here. Conversely, mixture of PPCA and Fisher-EM provide a better clustering
results both from a global point of view (respectively 89\% and 91\%
of clustering accuracy) and from a medical point of view since Fisher-EM
makes significantly less false negatives with an acceptable number
of false positives.

More importantly, Fisher-EM provides information which can be interpreted
\emph{a posteriori} to better understand both the data and the phenomenon.
Indeed, the values of the estimated loading matrix $U$, which is
a $6\,168\times1$ matrix here, expressed the correlation between
the discriminative subspace and the original variables. It is therefore
possible to identify the original variables with the highest power
of discrimination. It is important to highlight that Fisher-EM extracts
this information from the data in a unsupervised framework. Figure~\ref{fig:Maldi-2}
shows the correlation between each original variable and the discriminative
subspace on an arbitrary scale. The peaks of this curve correspond
to the original variables which have a high correlation with the discriminative
axis estimated by Fisher-EM. 

Figure~\ref{fig:Maldi-3} plots the difference between the mean spectra
of the classes cancer and control (cancer - control) and indicates
as well, using red triangles, the most discriminative original variables
(m/z values). It is not surprising to see that original variables
where the cancer and control spectra have a big difference are among
the most discriminative. More surprisingly, Fisher-EM selects the
original variables with m/z values equal to 2800 and 3050 as discriminative
variables whereas the difference between cancer and control spectra
is less for these variables than the difference on the variable with
m/z value equal to 1350. Such information, which have extracted from
the data in a unsupervised framework, may help the practitioner to
understand the clustering results.

\begin{figure}
\begin{centering}
\includegraphics[bb=0bp 300bp 576bp 525bp,clip,width=0.8\columnwidth]{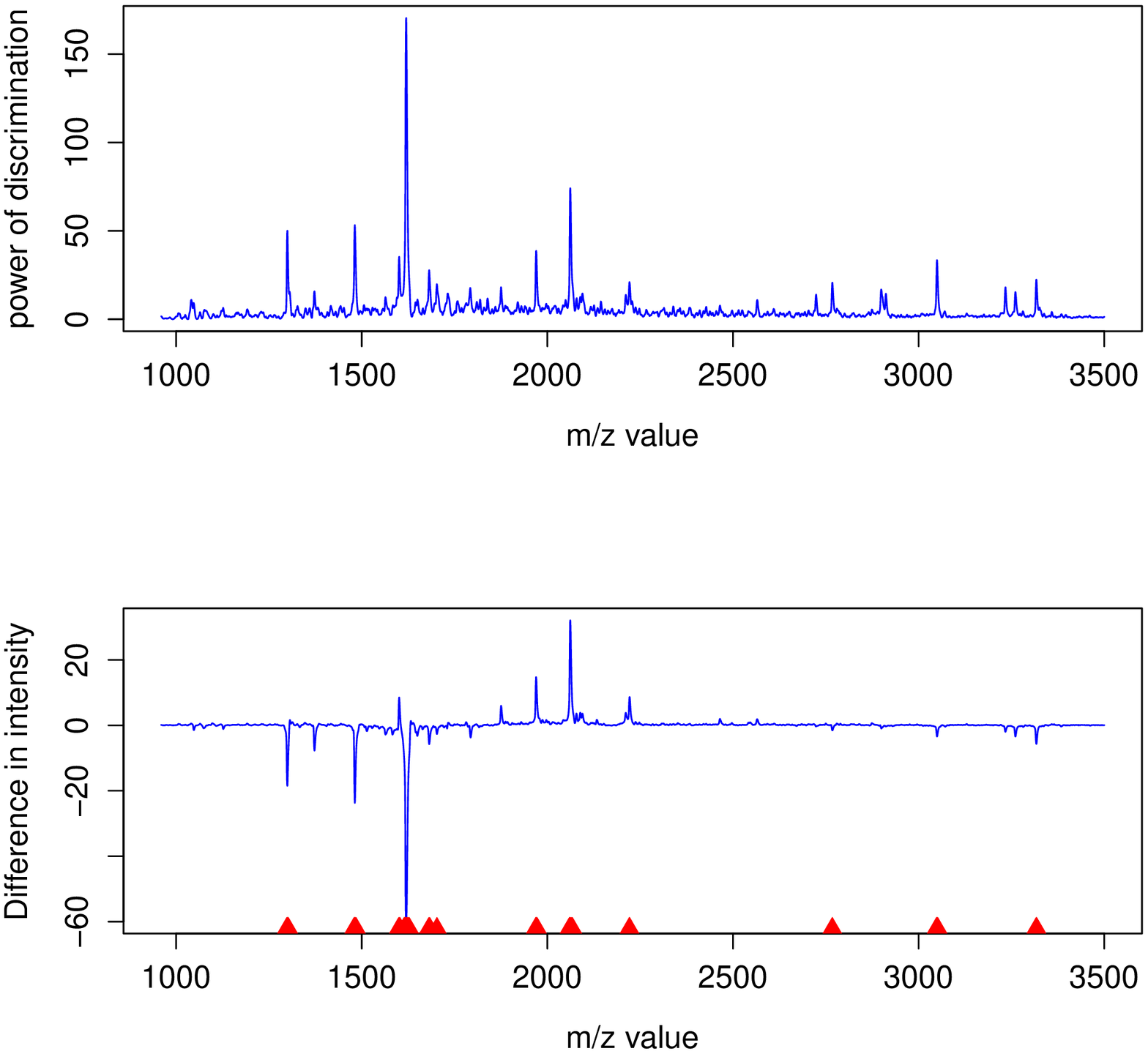}
\par\end{centering}

\caption{\label{fig:Maldi-2}Discrimination power of the original variables:
correlation between original variables and the discriminative subspace
on an arbitrary scale. }

\end{figure}

\begin{figure}
\begin{centering}
\includegraphics[bb=0bp 15bp 576bp 235bp,clip,width=0.8\columnwidth]{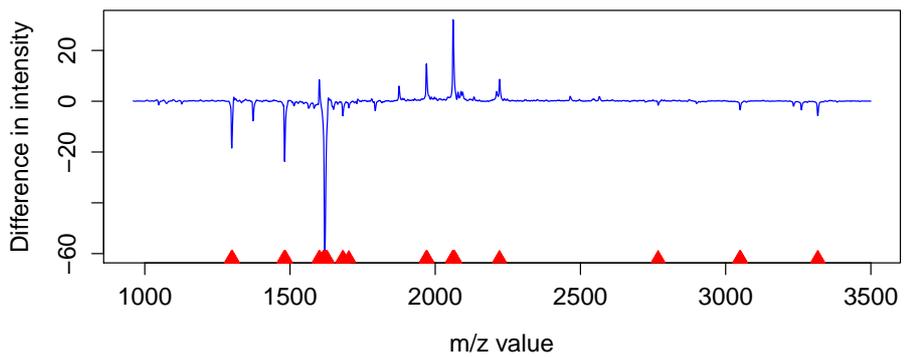}
\par\end{centering}

\caption{\label{fig:Maldi-3}Difference between the mean spectra of the classes
cancer and control (cancer - control) and most discriminative variables
(indicated by red triangles).}

\end{figure}

\section{Conclusion and further works\label{sec:Conclusion}}

This work has presented a discriminative latent mixture model which
models the data in a latent orthonormal discriminative subspace with
an intrinsic dimension lower than the dimension of the original space.
A family of 12 parsimonious DLM models has been exhibited by constraining
model parameters within and between groups. An estimation algorithm,
called the Fisher-EM algorithm, has been also proposed for estimating
both the mixture parameters and the latent discriminative subspace.
The determination procedure for the discriminative subspace adapts
the well-known Fisher criterion to the unsupervised classification
context under an orthonormality constraint. Furthermore, when the
number of groups is not too large, the estimated discriminative subspace
allows a useful projection of the clustered data. Experiments on simulated
and real datasets have shown that Fisher-EM performs better than existing
clustering methods. The Fisher-EM algorithm has been also applied
to the clustering of mass spectrometry data, which is a real-world
and complex application. In this specific context, Fisher-EM has shown
its ability to both efficiently cluster high-dimensional mass spectrometry
data and give a pertinent interpretation of the results.

However, the convergence of the Fisher-EM algorithm has been proved
in this work only for 2 of the DLM models and the convergence for
other models should be investigated. We feel that the convergence
could be proved for these models at least in a generalized EM context.
Among the other possible extensions of this work, it could be interesting
to find a way to visualize in 2D or 3D the clustered data when the
estimated discriminative subspace has more than 4 dimensions. Another
extension could be to consider a kernel version of Fisher-EM. For
this, it would be necessary to replace the Gram matrix introduced
in Section~\ref{par_n<<p} by a kernel. Finally, it could be also
interesting to introduce sparsity in the loading matrix through a
$\ell_{1}$ penalty in order to ease the interpretation of the discriminative
axes.

\section*{Acknowledgments}

The authors are indebted to the three referees and the editor for
their helpful comments and suggestions. They have contributed to greatly
improve this article.

\appendix

\section{Appendix}

In order not to surcharge the notations, the index $q$ of the current
iteration of the Fisher-EM algorithm is not indicated in the following
proofs. We also define the matrices $\tilde{W}$ and $\bar{W}$ such
that $W=\tilde{W}+\bar{W}$. The matrix $\tilde{W}$ is defined as
a $p\times p$ matrix containing the $d$ first vectors of $W$ completed
by zeros such as $\tilde{W}=[U,0_{p-d}]$ and $\bar{W}=W-\tilde{W}$
is defined by $\bar{W}=[0_{d},V]$.

\subsection{E step\label{a:costfunc}}

\begin{proof}[Proof of Proposition~\ref{prop:Estep}]The conditional
expectation $t_{ik}=E[P(z_{ik}|y_{i},\Theta)]$ can be viewed as well
as the posterior probability of the observation $y_{i}$ given a group
$k$ and, thanks to the Bayes' formula, can be written: \[
t_{ik}=\frac{\pi_{k}\phi(y_{i},\theta_{k})}{\sum_{l=1}^{K}\pi_{l}\phi(y_{i},\theta_{l})},\]
 where $\phi$ is the Gaussian density, and $\pi_{k}$ and $\theta_{k}$
are the parameters of the $k$th mixture component estimated in the
previous iteration. This posterior probability $t_{ik}$ can also
be formulated from the cost function $\Gamma_{k}$ such that: \[
t_{ik}=\frac{1}{\sum_{l=1}^{K}\exp\left(\frac{1}{2}(\Gamma_{k}(y_{i})-\Gamma_{l}(y_{i}))\right)},\]
 where $\Gamma_{k}(y_{i})=-2\log(\pi_{k}\phi(y_{i},\theta_{k}))$.
According to the assumptions of the model $\mathrm{DLM}_{[\Sigma_{k}\beta_{k}]}$
and given that $W=\tilde{W}+\bar{W}$, $\Gamma_{k}$ can be reformulated
as:\[
\begin{split}\Gamma_{k}(y_{i})= & \quad(y_{i}-m_{k})^{t}\tilde{W}\Delta_{k}^{-1}\tilde{W}^{t}(y_{i}-m_{k})+(y_{i}-m_{k})^{t}\bar{W}\Delta_{k}^{-1}\bar{W}^{t}(y_{i}-m_{k})\\
 & +\log(|\Delta_{k}|)-2\log(\pi_{k})+p~\log(2\pi),\end{split}
\]
 Moreover, since the relations $\tilde{W}(\tilde{W}^{t}\tilde{W})=\tilde{W}$
and $\bar{W}(\bar{W}^{t}\bar{W})=\bar{W}$ hold due to the construction
of $\tilde{W}$ and $\bar{W}$, then: \[
\begin{split}\Gamma_{k}(y_{i})= & \quad\left(\tilde{W}\tilde{W}^{t}(y_{i}-m_{k})\right)^{t}\tilde{W}\Delta_{k}^{-1}\tilde{W}^{t}\left(\tilde{W}\tilde{W}^{t}(y_{i}-m_{k})\right)\\
 & +\frac{1}{\beta_{k}}\left(\bar{W}\bar{W}^{t}(y_{i}-m_{k})\right)^{t}\left(\bar{W}\bar{W}^{t}(y_{i}-m_{k})\right)\\
 & +\log(|\Delta_{k}|)-2\log(\pi_{k})+p~\log(2\pi).\end{split}
\]
Let us now define $\vartheta_{k}=\tilde{W}\Delta_{k}^{-1}\tilde{W}^{t}$
and $||.||_{\vartheta_{k}}$, a norm on the latent space spanned by~$\tilde{W}$,
such that $||y||_{\vartheta_{k}}^{2}=y^{t}\vartheta_{k}y$. With these
notations, and according to the definition of $\Delta_{k}$, $\Gamma_{k}$
can be rewritten as: \[
\begin{split}\Gamma_{k}(y_{i})= & \quad||\tilde{W}\tilde{W}^{t}(y_{i}-m_{k})||_{\vartheta_{k}}^{2}+\frac{1}{\beta_{k}}||\bar{W}\bar{W}^{t}(y_{i}-m_{k})||^{2}\\
 & +\log(\left|\Sigma_{k}\right|)+(p-d)~\log(\beta_{k})-2\log(\pi_{k})+p~\log(2\pi).\end{split}
\]
Let us also define the projection operators $P$ and $P^{\perp}$
on the subspaces $\mathbb{E}$ and $\mathbb{E}^{\perp}$ respectively: 
\begin{itemize}
\item $P(y)=\tilde{W}\tilde{W}^{t}y$ is the projection of $y$ on the discriminative
space $\mathbb{E}$, 
\item $P^{\perp}(y)=\bar{W}\bar{W}^{t}y$ is the projection of $y$ on the
complementary space $\mathbb{E}^{\perp}$. 
\end{itemize}
Consequently, the cost function $\Gamma_{k}$ can be finally reformulated
as:\[
\begin{split}\Gamma_{k}(y_{i})= & \quad\left\Vert P(y_{i}-m_{k})\right\Vert {}_{\vartheta_{k}}^{2}+\frac{1}{\beta_{k}}\left\Vert P^{\perp}(y_{i}-m_{k})\right\Vert {}^{2}\\
 & +\log(\left|\Sigma_{k}\right|)+(p-d)~\log(\beta_{k})-2\log(\pi_{k})+p~\log(2\pi).\end{split}
\]
Since $P^{\perp}(y)=y-P(y)$, then the distance associated with the
complementary subspace can be rewritten as $||P^{\perp}(y_{i}-m_{k})||^{2}=||(y_{i}-m_{k})-P(y_{i}-m_{k})||^{2}$
and this allow to conclude. \end{proof}

\subsection{M step\label{a:estimation} }

\begin{proof}[Proof of Proposition~\ref{prop:Likelihood}]In the case
of the model $\mathrm{DLM}_{[\Sigma_{k}\beta_{k}]}$, at iteration
$q$, the conditional expectation of the complete log-likelihood $Q(y_{1},\dots,y_{n},\theta|\theta^{(q-1)})$
of the observed data $\{y_{1},\dots,y_{n}\}$ has the following form:

\begin{equation}
\begin{split}Q(\theta)= & \sum_{i=1}^{n}\sum_{k=1}^{K}t_{ik}\log(\pi_{k}\phi(y_{i},\theta_{k}))\\
= & \sum_{i=1}^{n}\sum_{k=1}^{K}t_{ik}\Bigl[-\frac{1}{2}\log(|S_{k}|)-\frac{1}{2}(y_{i}-m_{k})^{t}S_{k}^{-1}(y_{i}-m_{k})+\log(\pi_{k})-\frac{p}{2}~\log(2\pi)\Bigr],\end{split}
\label{eq:loglik_general}\end{equation}
 where $t_{ik}=E[z_{ik}|\theta^{(q-1)}]$. According to the definitions
of the diagonal matrix $\Delta_{k}$ and of the orientation matrix
$W$ for which $W^{-1}=W^{t}$, the inverse covariance matrix $S_{k}^{-1}$
of $Y$ can be written as $S_{k}^{-1}=(W\Delta_{k}W^{t})^{-1}=W^{-t}\Delta_{k}^{-1}W^{-1}=W\Delta_{k}^{-1}W^{t}$
and the determinant of $S_{k}$ can be also reformulated in the following
way: \[
|S_{k}|=|\Delta_{k}|=\left|\Sigma_{k}\right|\beta_{k}^{p-d}.\]
 Consequently, the complete log-likelihood $Q(\theta)$ can be rewritten
as: \begin{equation}
\begin{split}Q(\theta)= & -\frac{1}{2}\sum_{k=1}^{K}n_{k}\Bigl[-2\log(\pi_{k})+\log(\left|\Sigma_{k}\right|)+(p-d)~\log(\beta_{k})\\
 & +\frac{1}{n_{k}}\sum_{i=1}^{n}t_{ik}(y_{i}-m_{k})^{t}W\Delta_{k}^{-1}W^{t}(y_{i}-m_{k})+\gamma\Bigr].\end{split}
\label{e:loglik1}\end{equation}
where $n_{k}=\sum_{i=1}^{n}t_{ik}$ and $\gamma=p\log(2\pi)$ is a
constant term. At this point, two remarks can be done on the quantity
$\sum_{i=1}^{n}t_{ik}(y_{i}-m_{k})^{t}W\Delta_{k}^{-1}W^{t}(y_{i}-m_{k})$.
First, as this quantity is a scalar, it is equal to its trace. Secondly,
this quantity can be divided in two parts since $W=[U,V]$ and $W=\tilde{W}+\bar{W}$.
Then, the relation $W\Delta_{k}^{-1}W^{t}=\tilde{W}\Delta_{k}^{-1}\tilde{W}^{t}+\bar{W}\Delta_{k}^{-1}\bar{W}^{t}$
is stated and we can write: \[
\begin{array}{ccc}
(y_{i}-m_{k})^{t}W\Delta_{k}^{-1}W^{t}(y_{i}-m_{k}) & = & \mathrm{trace}\left((y_{i}-m_{k})^{t}\tilde{W}\Delta_{k}^{-1}\tilde{W}^{t}(y_{i}-m_{k})\right)\\
 &  & +\quad\mathrm{trace}\left((y_{i}-m_{k})^{t}\bar{W}\Delta_{k}^{-1}\bar{W}^{t}(y_{i}-m_{k})\right).\end{array}\]
Moreover, pointing out that $C_{k}=\frac{1}{n_{k}}\sum_{i=1}^{n}t_{ik}(y_{i}-m_{k})(y_{i}-m_{k})^{t}$
is the empirical covariance matrix the $k$th group, the previous
quantity can be rewritten as: \[
\frac{1}{n_{k}}\sum_{i=1}^{n}t_{ik}(y_{i}-m_{k})^{t}W\Delta_{k}^{-1}W^{t}(y_{i}-m_{k})=\mathrm{trace}(\Delta_{k}^{-1}\tilde{W}^{t}C_{k}\tilde{W})+\mathrm{trace}(\Delta_{k}^{-1}\bar{W}^{t}C_{k}\bar{W})\]
 and finally: \[
\frac{1}{n_{k}}\sum_{i=1}^{n}t_{ik}(y_{i}-m_{k})^{t}W\Delta_{k}^{-1}W^{t}(y_{i}-m_{k})=\quad\mathrm{trace}(\Sigma_{k}^{-1}U^{t}C_{k}U)\quad+\quad\sum_{j=1}^{p-d}\frac{v_{j}^{t}C_{k}v_{j}}{\beta_{k}},\]
where $v_{j}$, is the $j$th column vector of $V$. However, since
$\bar{W}=W-\tilde{W}$ and $W=[U,V]$, it is also possible to write:
\[
\begin{split}\frac{1}{\beta_{k}}\sum_{j=1}^{p-d}v_{j}^{t}C_{k}v_{j}= & \frac{1}{\beta_{k}}\left(\sum_{j=1}^{p}w_{j}^{t}C_{k}w_{j}-\sum_{j=1}^{d}u_{j}^{t}C_{k}u_{j}\right)\\
= & \frac{1}{\beta_{k}}\left(\sum_{j=1}^{p}\mathrm{trace}(w_{j}w_{j}^{t}C_{k})-\sum_{j=1}^{d}u_{j}^{t}C_{k}u_{j}\right)\\
= & \frac{1}{\beta_{k}}\Bigl[\mathrm{trace}(C_{k})-\sum_{j=1}^{d}u_{j}^{t}C_{k}u_{j}\Bigr].\end{split}
\]
Consequently, replacing this quantity in~(\ref{e:loglik1}) provides
the final expression of $Q(\theta)$. \end{proof}\medskip

\begin{proof}[Proof of Proposition~\ref{prop:MLestimates}]The maximization
of $Q(\theta)$ conduces for the DLM models to the following estimates.

\paragraph{Estimation of $\mathbf{\pi_{k}}$}

The prior probability $\pi_{k}$ of the group $k$ can be estimated
by maximizing $Q(\theta)$ with respect to the constraint $\sum_{k=1}^{K}\pi_{k}=1$
which is equivalent to maximize the Lagrange function: \[
L=Q(\theta)+\lambda\left(\sum_{k=1}^{K}\pi_{k}-1\right),\]
 where $\lambda$ is the Lagrange multiplier. Then, the partial derivative
of $L$ with respect to $\pi_{k}$ is $\frac{\partial L}{\partial\pi_{k}}=\frac{n_{k}}{\pi_{k}}+\lambda.$
Consequently: \[
\forall k=1,\dots,K\qquad\frac{\partial L}{\partial\pi_{k}}=0\Longleftrightarrow\frac{n_{k}}{\pi_{k}}+\lambda=0\Longleftrightarrow n_{k}+\lambda\pi_{k}=0,\]
 and: \[
\sum_{k=1}^{K}(n_{k}+\lambda\pi_{k})=n+\lambda=0\Longrightarrow\lambda=-n.\]
 Replacing $\lambda$ by its value in the partial derivative conduces
to an estimation of $\pi_{k}$ by: \[
\hat{\pi}_{k}=\frac{n_{k}}{n}.\]

\paragraph{Estimation of $\mathbf{\mu_{k}}$}

The mean $\mu_{k}$ of the $k$th group in the latent space can be
also estimated by maximizing the expectation of the complete log-likelihood
(equation \ref{eq:loglik_general}), which can be written in the following
way:\begin{equation}
Q(\theta)=\sum_{i=1}^{n}\sum_{k=1}^{K}t_{ik}\Bigl[-\frac{1}{2}\log(|S_{k}|)-\frac{1}{2}(y_{i}-U\mu_{k})^{t}S_{k}^{-1}(y_{i}-U\mu_{k})+\log(\pi_{k})-\frac{p}{2}\log(2\pi)\Bigr].\label{eq:loglik_mu_k}\end{equation}
 Consequently, the partial derivative of $Q$ with respect to $\mu_{k}$
is $\frac{\partial Q(\theta)}{\partial\mu_{k}}=-\frac{1}{2}\sum_{i=1}^{n}t_{ik}U^{t}(y_{i}-U\mu_{k}).$
Setting this quantity to $0$ gives: \[
\frac{\partial Q(\theta)}{\partial\mu_{k}}=0\Longleftrightarrow\sum_{i=1}^{n}t_{ik}U^{t}y_{i}=\sum_{i=1}^{n}t_{ik}\mu_{k}.\]
 and conduces to: \[
\hat{\mu}_{k}=\frac{1}{n_{k}}\sum_{i=1}^{n}t_{ik}U^{t}y_{i}.\]

\paragraph{Model $\mathrm{DLM}_{[\Sigma_{k}\beta_{k}]}$}

From Equation~(\ref{e:loglik2}), the partial derivative of $Q(\theta)$
with respect to $\Sigma_{k}$ has the following form:

\[
\frac{\partial Q(\theta)}{\partial\Sigma_{k}}=-\frac{n_{k}}{2}\frac{\partial}{\partial\Sigma_{k}}\left[\log(|\Sigma_{k}|)+\mathrm{trace}\left(\Sigma_{k}^{-1}U^{t}C_{k}U\right)\right].\]
Using the matrix derivative formula of the logarithm of a determinant,
$\frac{\partial\log(|A|)}{\partial A}=\left(A^{-1}\right)^{t}$, and
of the trace of a product, $\frac{\partial\mathrm{trace}(A^{-1}B)}{\partial A}=-\left(A^{-1}BA^{-1}\right)^{t}$,
the equality of $\frac{\partial Q(\theta)}{\partial\Sigma_{k}}$ to
the $d\times d$ zero matrix yields to the relation:\[
\Sigma_{k}^{-1}=\Sigma_{k}^{-1}U^{t}C_{k}U\Sigma_{k}^{-1},\]
and, by multiplying on the left and on the right by $\Sigma_{k},$
we find out the estimate of~$\Sigma_{k}$:\begin{equation}
\hat{\Sigma}_{k}=U^{t}C_{k}U.\label{eq:estim_Sigmak}\end{equation}
The estimation of $\beta_{k}$ is also obtained by maximizing $Q$
subject to $\beta_{k}$: \[
\frac{\partial Q(\theta)}{\beta_{k}}=0\Longleftrightarrow\frac{p-d}{\beta_{k}}-\frac{\mathrm{trace}(C_{k})}{\beta_{k}^{2}}+\frac{1}{\beta_{k}^{2}}\sum_{j=1}^{d}u_{j}^{t}C_{k}u_{j}=0,\]
and it is possible to conclude: \begin{equation}
\hat{\beta}_{k}=\frac{\mathrm{trace}(C_{k})-\sum_{j=1}^{d}u_{j}^{t}C_{k}u_{j}}{p-d}.\label{estim:Bk}\end{equation}

\paragraph{Model $\mathrm{DLM}_{[\Sigma_{k}\beta]}$ }

In this case, $Q$ has the following form:

\[
\begin{split}Q(\theta)= & -\frac{1}{2}\Bigl(\sum_{k=1}^{K}n_{k}\Bigl[-2\log(\pi_{k})+\mathrm{trace}(\Sigma_{k}^{-1}U^{t}C_{k}U)+\log(\left|\Sigma_{k}\right|)\Bigr]\\
 & +\sum_{k=1}^{K}n_{k}(p-d)~\log(\beta)+\sum_{k=1}^{K}\frac{n_{k}}{\beta}\Bigl[\mathrm{trace}(C_{k})-\sum_{j=1}^{d}u_{j}^{t}C_{k}u_{j}\Bigr]\Bigr),\\
= & -\frac{1}{2}\Bigl(\sum_{k=1}^{K}n_{k}\Bigl[-2\log(\pi_{k})+\mathrm{trace}(\Sigma_{k}^{-1}U^{t}C_{k}U)+\log(\left|\Sigma_{k}\right|)+\gamma\Bigr]\\
 & +n(p-d)~\log(\beta)+\frac{1}{\beta}\Bigl[n\mathrm{~trace}(C)-n\sum_{j=1}^{d}u_{j}^{t}Cu_{j}\Bigr]\Bigr),\end{split}
\]
where $C$ is the soft within covariance matrix of the whole dataset.
Setting to $0$ the partial derivative of $Q(\theta)$ conditionally
to $\beta$ implies $\frac{p-d}{\beta}-\frac{1}{\beta^{2}}\mathrm{trace}(C)+\frac{1}{\beta^{2}}\sum_{j=1}^{d}u_{j}^{t}Cu_{j}=0$
and this conduces to: \begin{equation}
\hat{\beta}=\frac{1}{p-d}\left(\mathrm{trace}(C)-\sum_{j=1}^{d}u_{j}^{t}Cu_{j}\right),\label{estim:B}\end{equation}
and the estimation of $\Sigma_{k}$ is given by Equation~(\ref{eq:estim_Sigmak}).

\paragraph{Model $\mathrm{DLM}_{[\Sigma\beta_{k}]}$}

The quantity $Q$ can be rewritten in this manner:

\[
\begin{split}Q(\theta)= & -\frac{1}{2}\Bigl(\sum_{k=1}^{K}n_{k}\Bigl[-2\log(\pi_{k})\Bigr]+n~\log(\left|\Sigma\right|)+n~\mathrm{trace}(\Sigma^{-1}U^{t}CU)\Bigr]\\
 & +\sum_{k=1}^{K}n_{k}\Bigl[(p-d)~\log(\beta_{k})+\frac{1}{\beta_{k}}\left(\mathrm{trace}(C_{k})-\sum_{j=1}^{d}u_{j}^{t}C_{k}u_{j}\right)+\gamma\Bigr]\Bigr),\end{split}
\]
then, the partial derivative of $Q(\theta)$ with respect to $\Sigma$
is: \[
\frac{\partial Q(\theta)}{\partial\Sigma}=-\frac{n}{2}\frac{\partial}{\partial\Sigma}\left[\log(|\Sigma|)+\mathrm{trace}\left(\Sigma^{-1}U^{t}CU\right)\right]\]
 and setting to $0$ provides the estimation of $\Sigma$:\begin{equation}
\hat{\Sigma}=U^{t}CU.\label{eq:estim_sigma}\end{equation}
Finally, the estimation of $\beta_{k}$ is provided by Equation~(\ref{estim:Bk}).

\paragraph{Model $\mathrm{DLM}_{[\Sigma\beta]}$}

The estimations of $\Sigma$ and $\beta$ have been already considered
above and are given by Equations~(\ref{eq:estim_sigma} and \ref{estim:B}).

\paragraph{Model $\mathrm{DLM}_{[\alpha_{kj}\beta_{k}]}$}

In this case, $Q$ has the following form: \[
Q(\theta)=-\frac{1}{2}\sum_{k=1}^{K}n_{k}\Bigl[-2\log(\pi_{k})+\sum_{j=1}^{d}\left(\log(\alpha_{kj})+\frac{u_{j}^{t}C_{k}u_{j}}{\alpha_{jk}}\right)+(p-d)~\log(\beta_{k})+\frac{1}{\beta_{k}}\sum_{j=d+1}^{p}v_{j}^{t}C_{k}v_{j}+\gamma\Bigr].\]
The partial derivative of $Q$ with respect to $\alpha_{kj}$ is $\frac{\partial Q(\theta)}{\partial\alpha_{kj}}=-\frac{1}{2n_{k}}\left(\frac{1}{\alpha_{kj}}-\frac{u_{j}^{t}C_{k}u_{j}}{\alpha_{kj}^{2}}\right)$
and setting to $0$ provides the estimate of $\alpha_{kj}$: \begin{equation}
\hat{\alpha}_{kj}=u_{j}^{t}C_{k}u_{j}.\label{estim:Akj}\end{equation}
The estimation of $\beta_{k}$ is provided by Equation~(\ref{estim:Bk}).

\paragraph{Model $\mathrm{DLM}_{[\alpha_{kj}\beta]}$}

The estimations of $\alpha_{kj}$ and $\beta$ have been already considered
above and are given by Equations~(\ref{estim:Akj} and \ref{estim:B}).

\paragraph{Model $\mathrm{DLM}_{[\alpha_{k}\beta_{k}]}$}

For this model, the expectation of the complete log-likelihood $Q(\theta)$
has the following form: \[
\begin{split}Q(\theta)= & -\frac{1}{2}\sum_{k=1}^{K}n_{k}\Bigl[-2\log(\pi_{k})+\sum_{j=1}^{d}\left(\log(\alpha_{k})+\frac{u_{j}^{t}C_{k}u_{j}}{\alpha_{k}}\right)\\
 & +(p-d)~\log(\beta_{k})+\frac{1}{\beta_{k}}\sum_{j=1}^{p-d}v_{j}^{t}C_{k}v_{j}+\gamma\Bigr],\\
Q(\theta)= & -\frac{1}{2}\sum_{k=1}^{K}n_{k}\Bigl[-2\log(\pi_{k})+d~\log(\alpha_{k})+\frac{1}{\alpha_{k}}\sum_{j=1}^{d}u_{j}^{t}C_{k}u_{j}\\
 & +(p-d)~\log(\beta_{k})+\frac{1}{\beta_{k}}\sum_{j=1}^{p-d}v_{j}^{t}C_{k}v_{j}+\gamma\Bigr].\end{split}
\]
 The partial derivative of $Q(\theta)$ with respect to $\alpha_{k}$
is $\frac{\partial Q(\theta)}{\partial\alpha_{k}}=-\frac{1}{2n_{k}}\left(\frac{d}{\alpha_{k}}-\frac{1}{\alpha_{k}^{2}}\sum_{j=1}^{d}u_{j}^{t}C_{k}u_{j}\right),$
and setting this quantity to $0$, provides: \begin{equation}
\hat{\alpha}_{k}=\frac{1}{d}\sum_{j=1}^{d}u_{j}^{t}C_{k}u_{j}.\label{estim:Ak}\end{equation}
On the other hand, the estimation of $\beta_{k}$ is the same as in
Equation~(\ref{estim:Bk}).

\paragraph{Model $\mathrm{DLM}_{[\alpha_{k}\beta]}$}

The estimations of $\alpha_{k}$ and $\beta$ are respectively provided
by Equations~(\ref{estim:Ak}) and~(\ref{estim:B}).

\paragraph{Model $\mathrm{DLM}_{[\alpha_{j}\beta_{k}]}$}

In this case, $Q(\theta)$ has the following form: \[
\begin{split}Q(\theta)= & -\frac{1}{2}\sum_{k=1}^{K}n_{k}\Bigl(-2\log(\pi_{k})+\sum_{j=1}^{d}\left(\log(\alpha_{j})+\frac{u_{j}^{t}C_{k}u_{j}}{\alpha_{j}}\right)\\
 & +(p-d)\log(\beta_{k})+\frac{1}{\beta_{k}}\sum_{j=1}^{p-d}v_{j}^{t}C_{k}v_{j}+\gamma\Bigr),\\
Q(\theta)= & -\frac{1}{2}\Bigl(\sum_{k=1}^{K}n_{k}\Bigl[-2\log(\pi_{k})\Bigr]+n\sum_{j=1}^{d}\log(\alpha_{j})+n\sum_{j=1}^{d}\frac{u_{j}^{t}Cu_{j}}{\alpha_{j}}\\
 & +\sum_{k=1}^{K}n_{k}\Bigl[(p-d)~\log(\beta_{k})+\frac{1}{\beta_{k}}\sum_{j=1}^{p-d}v_{j}^{t}C_{k}v_{j}+\gamma\Bigr]\Bigr).\end{split}
\]
 The partial derivative of $Q(\theta)$ with respect to $\alpha_{j}$
is $\frac{\partial Q(\theta)}{\partial\alpha_{j}}=-\frac{n}{2}\left(\frac{1}{\alpha_{j}}-\frac{1}{\alpha_{j}^{2}}u_{j}^{t}Cu_{j}\right)$
and setting to $0$ implies: \begin{equation}
\hat{\alpha}_{j}=u_{j}^{t}Cu_{j},.\label{estim:Aj}\end{equation}
and the estimation of $\beta_{k}$ is the same as in Equation~(\ref{estim:Bk}).

\paragraph{Model $\mathrm{DLM}_{[\alpha_{j}\beta]}$}

The estimations of $\alpha_{j}$ and $\beta$ are respectively provided
by Equations~(\ref{estim:Aj}) and~(\ref{estim:B}).

\paragraph{Model $\mathrm{DLM}_{[\alpha\beta_{k}]}$}

In this case, $Q(\theta)$ has the following form:\[
\begin{split}Q(\theta)= & -\frac{1}{2}\sum_{k=1}^{K}n_{k}\Bigl(-2\log(\pi_{k})+d~\log(\alpha)+\frac{1}{\alpha}\sum_{j=1}^{d}u_{j}^{t}C_{k}u_{j}+\\
 & +(p-d)\log(\beta_{k})+\frac{1}{\beta_{k}}\sum_{j=1}^{p-d}v_{j}^{t}C_{k}v_{j}+\gamma\Bigr),\\
Q(\theta)= & -\frac{1}{2}\Bigl(\sum_{k=1}^{K}n_{k}[-2\log(\pi_{k})]+n~d~\log(\alpha)+\frac{n}{\alpha}\sum_{j=1}^{d}u_{j}^{t}Cu_{j}\\
 & +\sum_{k=1}^{K}n_{k}\Bigl[(p-d)\log(\beta_{k})+\frac{1}{\beta_{k}}\sum_{j=1}^{p-d}v_{j}^{t}C_{k}v_{j}+\gamma\Bigr]\Bigr),\end{split}
\]
 The partial derivative of $Q(\theta)$ with respect to $\alpha$
is $\frac{\partial Q(\theta)}{\partial\alpha}=-\frac{n}{2}\left(\frac{d}{\alpha}-\frac{1}{\alpha^{2}}\sum_{j=1}^{d}u_{j}^{t}Cu_{j}\right),$
and setting this quantity to $0$, we end up with: \begin{equation}
\hat{\alpha}=\frac{1}{d}\sum_{j=1}^{d}u_{j}^{t}Cu_{j}.\label{estim:A}\end{equation}
 The estimation of $\beta_{k}$ is the same as in Equation~(\ref{estim:Bk}).

\paragraph{Model $\mathrm{DLM}_{[\alpha\beta]}$}

The estimations of $\alpha$ and $\beta$ have been already computed
and are provided by Equations~(\ref{estim:A}) and~(\ref{estim:B}).\end{proof}

\bibliographystyle{plain}
\bibliography{bibliFEM}

\end{document}